\documentclass[10pt, conference]{IEEEtran}
\IEEEoverridecommandlockouts

\usepackage[algoruled,vlined,boxed,linesnumbered]{algorithm2e}
\usepackage{amsmath,amssymb,amsfonts}
\usepackage{algorithmic}
\usepackage{graphicx}
\usepackage{textcomp}
\usepackage{xcolor}
\usepackage{pgfplots}
\usepackage{listings}
\usepackage{soul}
\usepackage{amsthm}
\usepackage{colortbl}
\usepackage[numbers,sort&compress]{natbib}

\pgfplotsset{compat=1.8}

\newcommand{\mK}{ \overline{ \big(\begin{smallmatrix} m_i \\ K_i \end{smallmatrix}\big)}}
\newcommand{\wh}{ \overline{ \big(\begin{smallmatrix} w_i \\ w_i+h_i \end{smallmatrix}\big)}}
\newcommand{\wharg}[2]{ \overline{ \big(\begin{smallmatrix} #1 \\ #2 \end{smallmatrix}\big)}}
\newcommand{\hw}{ (\begin{smallmatrix} h_i \\ w_i+h_i \end{smallmatrix})}

\def\BibTeX{{\rm B\kern-.05em{\sc i\kern-.025em b}\kern-.08em
    T\kern-.1667em\lower.7ex\hbox{E}\kern-.125emX}}

\newtheorem{definition}{Definition}
\newtheorem{lemma}{Lemma}
\newtheorem{theorem}{Theorem}
\newtheorem{corollary}{Corollary}

\begin{document}

\title{Global Scheduling of Weakly-Hard Real-Time Tasks using Job-Level Priority Classes}

\newcommand{\linebreakand}{
    \end{@IEEEauthorhalign}
    \hfill\mbox{}\par
    \mbox{}\hfill\begin{@IEEEauthorhalign}
}

\author{
\IEEEauthorblockN{V. Gabriel Moyano}
\IEEEauthorblockA{
    \textit{Institute of Software Technology}\\
    \textit{German Aerospace Center (DLR)}\\
    Braunschweig, Germany \\
    gabriel.moyano@dlr.de}
\and
\IEEEauthorblockN{Zain A. H. Hammadeh}
\IEEEauthorblockA{
    \textit{Institute of Software Technology}\\
    \textit{German Aerospace Center (DLR)}\\
    Braunschweig, Germany \\
    zain.hajhammadeh@dlr.de}
\and
\IEEEauthorblockN{Selma Saidi}
\IEEEauthorblockA{
    \textit{Institute of Computer and Network Engineering}\\
    \textit{Technische Universität Braunschweig}\\
    Braunschweig, Germany \\
    saidi@ida.ing.tu-bs.de}
\and
\IEEEauthorblockN{Daniel L\"{u}dtke}
\IEEEauthorblockA{
    \centerline{Institute of Software Technology}\\
    \textit{German Aerospace Center (DLR)}\\
    Braunschweig, Germany \\
    daniel.luedtke@dlr.de}
}

\maketitle

\begin{abstract}
Real-time systems are intrinsic components of many pivotal applications, such as self-driving vehicles, aerospace and defense systems.
The trend in these applications is to incorporate multiple tasks onto fewer, more powerful hardware platforms, e.g., multi-core systems, mainly for reducing cost and power consumption.
Many real-time tasks, like control tasks, can tolerate occasional deadline misses due to robust algorithms.
These tasks can be modeled using the weakly-hard model.
Literature shows that leveraging the weakly-hard model can relax the over-provisioning associated with designed real-time systems.
However, a wide-range of the research focuses on single-core platforms.
Therefore, we strive to extend the state-of-the-art of scheduling weakly-hard real-time tasks to multi-core platforms.
We present a global job-level fixed priority scheduling algorithm together with its schedulability analysis.
The scheduling algorithm leverages the tolerable continuous deadline misses to assigning priorities to jobs.
The proposed analysis extends the Response Time Analysis (RTA) for global scheduling to test the schedulability of tasks.
Hence, our analysis scales with the number of tasks and number of cores because, unlike literature, it depends neither on Integer Linear Programming nor reachability trees.
Schedulability analyses show that the schedulability ratio is improved by 40\% comparing to the global Rate Monotonic (RM) scheduling  and up to 60\% more than the global EDF scheduling, which are the state-of-the-art schedulers on the RTEMS real-time operating system.
Our evaluation on industrial embedded multi-core platform running RTEMS shows that the scheduling overhead of our proposal does not exceed 60 Nanosecond.

\end{abstract}


\section{Introduction}
Enabling more autonomy in the automotive and the aerospace domains require involving sophisticated control algorithms with real-time requirements.
Computing hard real-time guarantees for the developed tasks under worst-case scenarios comes at the cost of exacerbating over-provisioning in such new software-based embedded systems, which implies, for example, higher power consumption.
Literature \cite{hammadeh2017budgeting} shows that leveraging the weakly-hard model can relax the over-provisioning associated with designed real-time systems.
Furthermore, many papers, e.g.~\cite{Pazzaglia:ECRTS2018, maggio2020control, vreman2021control}, proved that the control systems can tolerate occasional deadline misses with a small amount of performance degradation.
The weakly-hard real-time model~\cite{bernat2001weakly} extends the tight region of schedulable tasks that is defined by the hard real-time model by exploiting the tolerable deadline misses.
In the weakly-hard real-time model, the notation $\mK$ defines the maximum number of tolerable deadline misses $m_i$ in a sequence of $K_i$ releases.
In the last few years, weakly-hard real-time systems received a lot of attention, and different schedulability analysis have been proposed~\cite{xu2015TWCA, sun2017MILP, choi2021toward}.
To compute weakly-hard real-time guarantees, the developed analysis should not only consider the job(s) in the worst-case scenario, but also should consider all possible combinations of jobs within a window of $K$ consecutive jobs.
That makes computing the weakly-hard real-time guarantees more complicated and subject to more pessimism.

\begin{figure}[t!]
    \centering
        \resizebox{1\columnwidth}{!}{
            \includegraphics{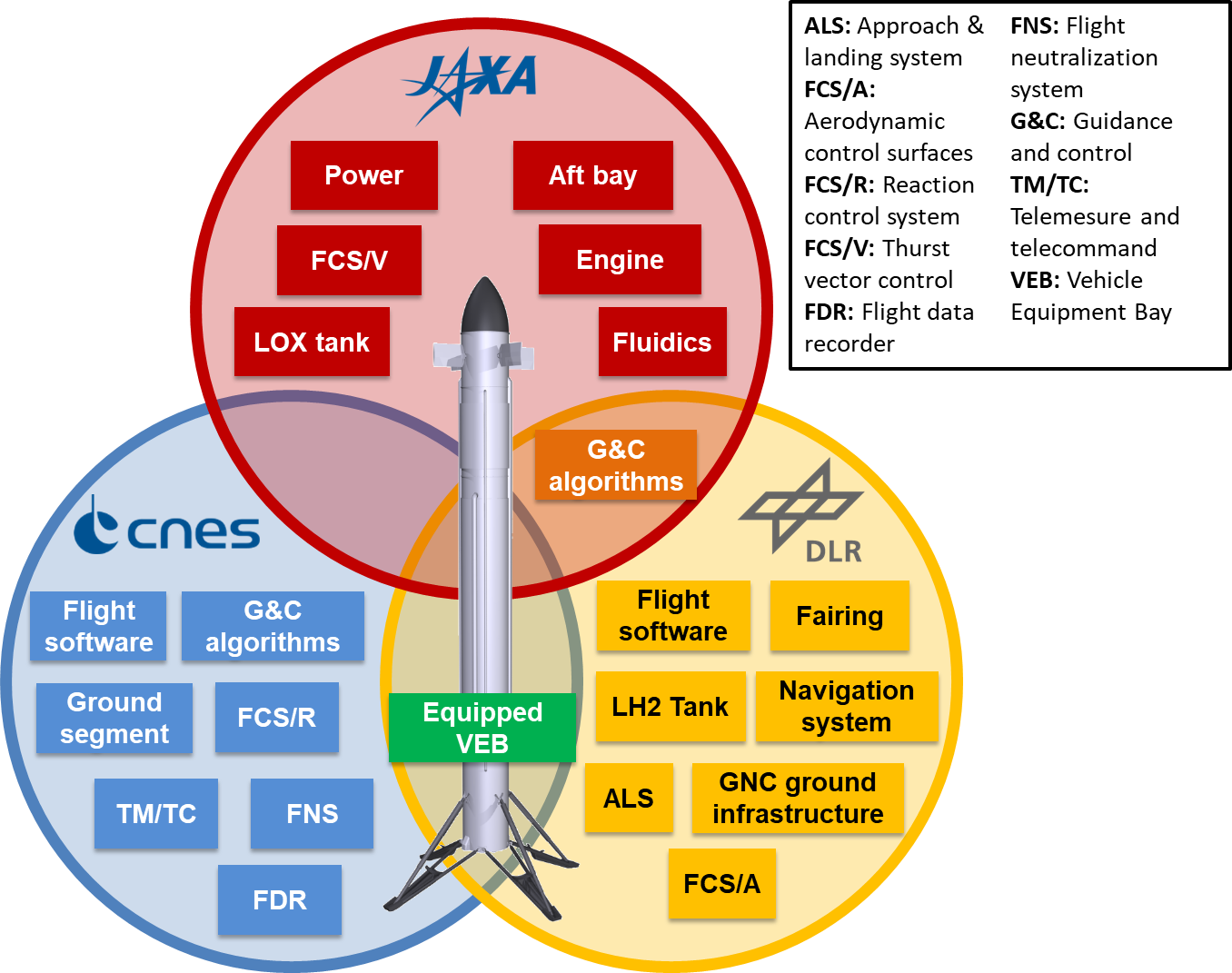}
        }
    \caption{CALLISTO project.}
    \label{fig:callisto}
    \vspace*{-8mm}
\end{figure}

Embedded system developers head for multi-core platforms in automotive and aerospace domains to fulfill the increasing demand for computational performance.
This is the case for the on-board computers in CALLISTO (Cooperative Action Leading to Launcher Innovation for Stage Tossback Operation) ~\cite{krummen2021towards, illig2022callisto};
a joint project between the French National Center for Space Studies (CNES), the German Aerospace Center (DLR) and the Japan Aerospace Exploration Agency (JAXA) for developing and building a Vertical-Takeoff/Vertical-Landing rocket that can be reused (\figurename~\ref{fig:callisto}).
The on-board computers in CALLISTO carry-out the state estimation, based on fusing information of several sensors~\cite{schwarz2019preliminary}, and the execution of the control algorithms.
These algorithms are complex and computationally demanding but also deemed to be robust enough to an occasional deadline miss, in which case the performance degradation is negligible.
Therefore, the weakly-hard model can be used to describe the tolerance of those algorithms.
However, multi-core systems were out of scope in the majority of papers that address weakly-hard real-time systems.
In this work, we propose a global scheduling algorithm that exploits the tolerable deadline misses, and a schedulability analysis to compute weakly-hard real-time guarantees.
The weakly-hard real-time constraints define the distribution of the tolerable deadline misses and deadline hits in a window of $K$ consecutive jobs.
Satisfying such constraints presuppose the job-level schedulability as a more efficient approach than task-level schedulability.
Hence, we propose in this work a global job-level fixed priority scheduling.

Choi et al. proposed in~\cite{choi2019jobclass, choi2021toward} a single-core job-level fixed priority scheduling, in which the jobs receive different priorities upon meeting/missing their deadlines.
When a job meets its deadline, the next job is assigned a lower priority from a predefined job-level priority classes.
When a deadline miss occurs, the next job will get a higher or the highest priority possible in the predefined priority class.
Choi et al. proposed a reachability tree-based analysis, and they proposed a heuristic to extend their work to semi-partitioned multi-core scheduling.

This work recalls the job-class level scheduling to propose a multi-core global scheduling. Our main contributions are as follows:
\begin{itemize}
	\item We prove that satisfying the weakly-hard constraint $\wh$ for jobs that have the highest priority in the priority classes is sufficient to guarantee that the task $\tau_i$ satisfies the constraint $\mK$, where $w_i$ is the maximum number of tolerable consecutive deadline misses and $h_i$ is the minimum required deadline hits after $w_i$
	\item We propose a global job-level fixed priority scheduling using predefined priority classes for tasks that can tolerate a bounded number of deadline misses
	\item We propose a schedulability analysis for the proposed global scheduling.
    Our schedulability analysis extends the Response Time Analysis (RTA)~\cite{bertogna2007response}
\end{itemize}    
In our proposed analysis, we compromise between pessimism and complexity.
Using the constraint $\wh$ instead of $\mK$ limits the number of deadline sequences that satisfy the weakly-hard constraint.
However, we need neither a reachability tree-based analysis nor an Integer Linear Programming (ILP) based analysis, e.g.,~\cite{sun2017MILP}.
Therefore, our analysis has less complexity than~\cite{choi2019jobclass, choi2021toward} and scales with the number of tasks because we use RTA~\cite{bertogna2007response}.
The proposed scheduling can be seen as a fault-tolerance mechanism in which the deadline-miss is an error and assigning a higher priority to the next job as a mitigation mechanism.
However, the problem will be studied from real-time schedulability perspective.

The rest of this paper is organized as follows: the next section recalls the related work.
In Section~\ref{sec:system-model}, we present our system model and we elaborate our problem statement.
Section~\ref{sec:RTA} recalls the Response Time Analysis (RTA).
Our contribution starts in Section~\ref{sec:contribution-algorithm} by showing the global scheduling algorithm for weakly-hard real-time tasks.
Then, the analysis for the presented scheduling algorithm is shown in Section~\ref{sec:contribution-analysis}.
Also, we experimentally evaluate our proposed scheduling and we present the results in Section~\ref{sec:evaluation}.
Finally, Section~\ref{sec:conclusion} concludes our paper.

\section{Related Work}\label{sec:related}

Many global scheduling algorithms have been proposed as extensions of single-core scheduling algorithms, e.g., EDF (G-EDF) and Fixed Priority (G-FP)~\cite{bertogna2007response}.
The Pfair~\cite{Baruah2005ProportionatePA, Anderson-Pfair-2000} sub-category of global scheduling algorithms has been proven to be optimal for scheduling periodic and sporadic tasks with implicit deadlines. 
The main drawback of the Pfair global scheduling algorithms is the high number of preemptions, i.e., context switches, hence, high scheduling overhead~\cite{app131810131}. 
Nelissen et al. presented in~\cite{Nelissen-UEDF2012} the U-EDF algorithm, which is an unfair but optimal EDF global scheduling. 
U-EDF, among others, tried to mitigate the drawbacks of Pfair algorithms.
However, the non-Pfair optimal scheduling algorithms, e.g., U-EDF, have high computation complexity~\cite{app131810131}. 

In this paper, we aim to leverage the tolerable deadline misses to enhance the schedulability ratio of G-FP. 
Weakly-hard real-time constraints define the maximum number of deadline misses that a task can tolerate before going into a faulty state. 
The term weakly-hard was coined by Bernat et al. in~\cite{bernat2001weakly} to describe systems in which tasks have the $\overline{ (\begin{smallmatrix} m \\ K \end{smallmatrix})}$ constraints where $m$ represents the maximum number of tolerable deadline misses in a sequence of $K$ jobs. 
The notation $\overline{ (\begin{smallmatrix} m \\ K \end{smallmatrix})}$, though, is a bit older and was defined in \cite{hamdaoui1995dynamic} by Hamdaoui et al as $(m, K)$-firm. Since 2014, the number of papers addressing the weakly-hard real-time systems has increased significantly.

Sun et al. presented in \cite{sun2017MILP} a weakly-hard schedulability analysis that computes the maximum bound on $m$ within a time window of $K$ consecutive jobs using a Mixed Integer Linear Programming (MILP).
The MILP checks all possible scenarios within a time window of $K$ consecutive jobs where tasks are periodically activated. 
The analysis in \cite{sun2017MILP} can, therefore, provide tight bounds on $m$ with reasonable complexity for small $K \leq 10$ \cite{Natale-ESWEEK17}.

A Linear Programming (LP) based weakly-hard schedulability analysis has been presented in \cite{xu2015TWCA} for overloaded systems.
This approach considers temporarily overloaded systems due to rare sporadic jobs and bounds the impact of the sporadic overload jobs on the tasks, which are assumed to be schedulable in the non-overloaded intervals, in terms of deadline misses.
This approach has two features: 1) It scales with $K$ and number of tasks because it depends on an LP relaxation.
2) It is extendable for more scheduling policies. 
However, this approach reports a high pessimism for small $K$~\cite{hammadeh2019TWCA}.

Pazzaglia et al.~\cite{Pazzaglia:ECRTS2018} researched the performance cost of deadline misses in control systems.
They have shown the impact of the distribution of deadline misses within the sequence of $K$ jobs on the performance. 
Liang et al.~\cite{liang2020FualtWH} presented a fault tolerance mechanism for weakly-hard.

The job-class-level scheduling presented in~\cite{choi2019jobclass} and~\cite{choi2021toward} recalled the original concept proposed by Hamdaoui et al.~\cite{hamdaoui1995dynamic}, in which each task is assigned a different priority upon meeting/missing their deadlines.
The proposed scheduling in~\cite{choi2019jobclass, choi2021toward} is dedicated to single-core systems.
The author showed how to extend the job-class-level scheduling to semi-partitioned multi-core scheduling.
Our work extends the job-class-level scheduling to global multi-core scheduling.
However, our schedulability analysis does not depend on a reachability tree as the one in~\cite{choi2019jobclass, choi2021toward}.

Recently, Maggio et al. proposed in~\cite{maggio2020control, vreman2021control, Vreman2022control, Vreman2022jlWH} an approach to analyze the stability of control systems under different patterns of deadline misses.
The proposed approach by Maggio et al. can help in extracting the weakly-hard constraints, i.e., bounding $m$ and $K$.
The authors considered a system model of single-core platforms and periodic control tasks.

Wu and Jin proposed in~\cite{Wu2008} a global scheduling algorithm for multimedia streams.
They applied the Distance Based Priority (DBP) algorithm~\cite{hamdaoui1995dynamic} to a global scheduler where the task that is close to violate its $\mK$ constraint is assigned dynamically the highest priority.
In \cite{Kong2011}, Kong and Cho computed bounds on the probability of not satisfying the $\mK$ constraint and they proposed a dynamic  hierarchical scheduling algorithm to improve the quality of service.
The goal of our paper is different from \cite{Wu2008} and \cite{Kong2011} because we aim to exploit the weakly-hard constraints to increase the load that can be scheduled to a multi-core system under fixed priority scheduling.

\section{System Model}
\label{sec:system-model}
This paper considers independent sporadic tasks with constrained deadlines and preemptive scheduling.
A task set is executed on a Symmetric Multi-Processing (SMP) multi-core platform.
At the end of this section, Table~\ref{tab:system-model-variables} shows the notations used in this work.

{\bf Task model.}
A task $\tau_i$ is described using 5 parameters:
\begin{equation*}
    \tau_i \doteq (C_i, D_i, T_i, \mK)
\end{equation*}

\begin{itemize}
    \item $C_i$: The worst-case execution time of $\tau_i$.
    \item $D_i$: The relative deadline of each job of $\tau_i$. Since tasks have a constrained deadline $D_i \leq T_i$.
    \item $T_i$: The minimum inter-arrival time between consecutive jobs of $\tau_i$.
    \item $\mK$: The weakly-hard constraint of $\tau_i$, where $m_i$ is the number of tolerable deadline misses in a $K_i$ window, where $m_i < K_i$ and $m_i \geq 1$.
    A hard real-time task is characterized by $m_i = 0$ and $K_i = 1$.
\end{itemize}

In this paper, we use similar weakly-hard constraint notations as in \cite{bernat2001weakly}, see Table~\ref{tab:weakly-hard-notations}.
\begin{table}[t!]
    \centering
    \caption{Weakly-hard constraint notations}
        \resizebox{.8\columnwidth}{!}{
            \begin{tabular}{ c c c }
                \hline
                 & deadline hits & deadline misses \\
                any order & $\left(\begin{smallmatrix} m_i \\ K_i \end{smallmatrix}\right)$ & $\overline{\left( \begin{smallmatrix} m_i \\ K_i \end{smallmatrix} \right)}$ \\
                \rule{0pt}{4ex}
                consecutive & $\left<\begin{smallmatrix} m_i \\ K_i \end{smallmatrix}\right>$ & $\overline{\left< \begin{smallmatrix} m_i \\ K_i \end{smallmatrix} \right>}$ \\
                \hline
            \end{tabular}
        }
    \label{tab:weakly-hard-notations}
\end{table}

We classify the tasks based on the deadline misses tolerable in a $K_i$ window:
\begin{definition}
    Low-tolerance tasks: weakly-hard real-time tasks which require more deadline hits than tolerable misses in the $K_i$ window, i.e.\ tasks with a ratio $m_i/K_i < 0.5$ and $m_i > 0$.
    \label{def:low-tolerance-tasks}
\end{definition}
\begin{definition}
    High-tolerance tasks: weakly-hard real-time tasks which tolerate a bigger or equal quantity of deadline misses than quantity of deadline hits in the $K_i$ window, i.e.\ tasks with a ratio $m_i/K_i \geq 0.5$.
    \label{def:high-tolerance-tasks}
\end{definition}

{\bf Schedulability of weakly-hard tasks.}
\begin{definition}
    A weakly-hard task $\tau_i$ with constraint $\mK$ is schedulable if, in any window of $K_i$ consecutive invocations of the task, no more than $m_i$ deadlines are missed.
\end{definition}
\begin{definition}
	A deadline sequence is a binary sequence of length $K_i$, in which 1 represents a deadline hit and 0 represents a deadline miss.
\end{definition}

{\bf Utilization.}
The utilization of a task $\tau_i$ is defined as the fraction of processor time required by its execution:
\begin{equation*}
    U_i = \frac{C_i}{T_i}
\end{equation*}

Then, the utilization of the task set (also known as total utilization) is defined as the sum of all task utilizations:
\begin{equation*}
    U = \sum^{n}_{i = 1} U_i = \sum^{n}_{i = 1} \frac{C_i}{T_i}
\end{equation*}
where $n$ is the number of tasks in the task set.

{\bf System-level action for missed deadlines.}
The proposed scheduling algorithm and schedulability analysis considers the {\it Job-Kill} in case of a deadline miss.
In this system-level action, the job that does not meet its deadline is killed to remove load from the processor.

\subsection{Problem Statement}
In this work, we aim to exploit the weakly-hard constraints for increasing the number of schedulable tasks on an SMP multi-core platform.
Given a task set of independent weakly-hard tasks and an SMP multi-core platform,  our goal is to provide a global scheduling algorithm for the weakly-hard tasks and a scheduling analysis.

\begin{table}[h!]
    \centering
    \caption{Key mathematical notations used in this work.}
        \resizebox{\columnwidth}{!}{
            \begin{tabular}{ c | p{0.4\textwidth} }
                \hline
                Notation & Description \\
                \hline
                $n_c$ & number of cores in the system \\
                 \rowcolor{gray!10}$w_i$ & maximum consecutive number of deadline misses to uniformly distribute $m_i$ in a window of $K_i$ \\
                $h_i$ & deadline hits required per deadline miss \\
                 \rowcolor{gray!10}$\mathcal{JC}^{q}_{i}$ & job-class $q$ of task $\tau_i$ \\
                $jl_i$ & job-level variable of task $\tau_i$ \\
                 \rowcolor{gray!10}$s_i$ & slack of a task $\tau_i$ \\
                $N_i(L)$ & number of jobs interfering in the time interval $L$ \\
                 \rowcolor{gray!10}$O_i(L)$ & number of jobs not interfering in the time interval $L$ \\
                $a_i$ & variable for counting or not the interference of the carry-out job \\
                \hline
            \end{tabular}
        }
    \label{tab:system-model-variables}
\end{table}

\section{Original Response Time Analysis}
\label{sec:RTA}
Our paper extends the well-established approach Response Time Analysis (RTA)~\cite{bertogna2007response} in Section~\ref{sec:contribution-analysis}.
Therefore, we recall it in this section.
The response time $R_k$ of a task $\tau_k$ is defined as:
\begin{equation*}
R_k \doteq  C_k + I_k
\end{equation*}

Where $I_k$ is the upper bound on the interference from other tasks and it is computed as follows:
\begin{equation*}
I_k = \frac{1}{n_c} \sum_{i \neq k} I_{i,k}
\end{equation*}
Where $n_c$ is the number of cores in the system and $I_{i,k}$ is the interference of a task $\tau_i$ over the task $\tau_k$.
For Task-Level Fixed Priority (TLFP), only the tasks with higher priority than $\tau_k$ interfere, therefore, $I_k$ is reduced to:
\begin{equation*}
I_k = \frac{1}{n_c} \sum_{i \in hp(k)} I_{i,k}
\end{equation*}
Where $hp(k)$ is the set of tasks indices that have higher priority than $\tau_k$.

An upper bound on the response time $R^{ub}_{k}$ of a task $\tau_k$ can be computed by bounding the interference $I_k$.
For computing $I_{i,k}$, we bound the workload $\hat{W_i}$ imposed from $\tau_i$:
\begin{equation*}
    \hat{W}_i(L) = N_i(L) C_i + \min(C_i, (L - R^{ub}_i - C_i) \mod T_i)
\end{equation*}
Where $N_i(L)$ is the maximum number of jobs of $\tau_i$ that may execute within the time window of size $L$:
\begin{equation*}
N_i(L) = \left \lfloor \frac{L + R^{ub}_i - C_i}{T_i} \right \rfloor
\end{equation*}
$R^{ub}_i$ is the upper bound on the response time of $\tau_i$, which has a higher priority than $\tau_k$.
Hence, RTA iterates over the tasks in priority order to compute upper bounds on their response times.

$\hat{W}_i(L)$ calculates the workload imposed by $\tau_i$ considering the carry-in job, the body jobs and the carry-out job defined as follow (see \figurename~\ref{fig:rta-jobs}):
\begin{definition}
    A carry-in job is a job with a deadline within the interval of interest, but its release time is outside of it.
\end{definition}
\begin{definition}
    Body jobs are jobs with both their release time and deadline within the interval of interest.
\end{definition}
\begin{definition}
   A carry-out job is a job with a release time within the interval of interest but a deadline outside of it.
\end{definition}
To conservatively bound the workload of $\tau_i$ within the interval of interest $L$, we have to consider with the execution of the carry-in job.
Therefore, the first term of $\hat{W}_i(L)$ represents the workload due to the carry-in job and the body jobs while the second term bounds the workload due to the carry-out job.

\begin{figure}
    \centering
        \resizebox{1\columnwidth}{!}{
            \includegraphics{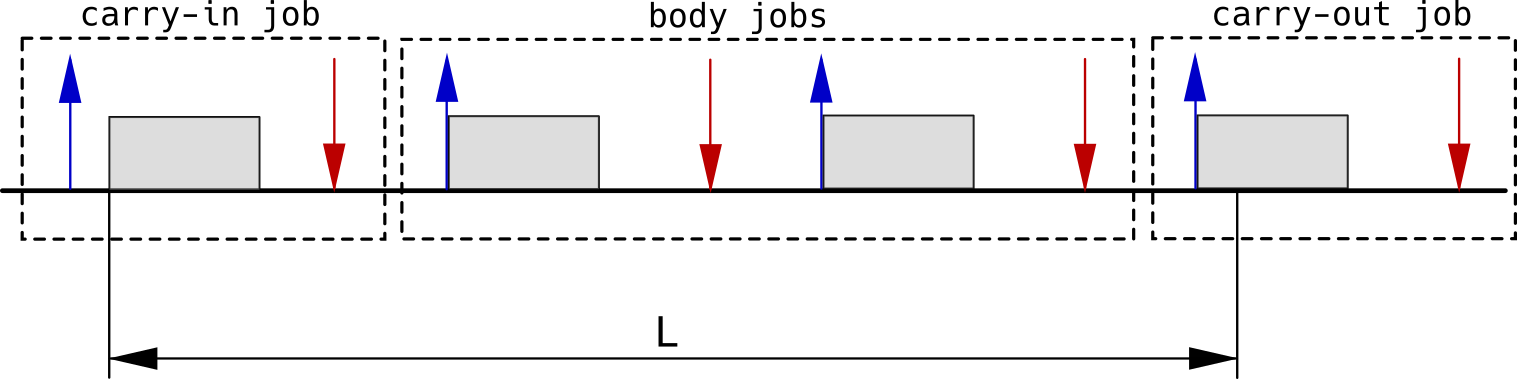}
        }
    \caption{Types of jobs within a time interval $L$ as proposed in RTA~\cite{bertogna2007response}.
    Blue arrows represent task activation, while red ones represent deadlines.
    Gray boxes refer to job executions.}
    \label{fig:rta-jobs}
\end{figure}

The bound on $\hat{W}_i$ and  $N_i(L)$ can be tightened by introducing the {\it slack} of $\tau_i$.
The slack of $\tau_i$ is calculated based on its response time as follows: $s_i = max(D_i - R^{ub}_{i}, 0)$.

Considering $s_i$, $\hat{W}_i$ and $N_i(L)$ are computed as follow:
\begin{equation}
    \label{eq:rta-bounded-workload}
    \hat{W}_i = N_i(L) C_i + \min(C_i, (L - R^{ub}_i - C_i - s_i) \mod T_i)
\end{equation}
\begin{equation}
    \label{eq:rta-number-jobs}
    N_i(L) = \left \lfloor \frac{L + R^{ub}_i - C_i - s_i}{T_i} \right \rfloor
\end{equation}
For TLFP, $R^{ub}_{k}$ of the task $\tau_k$ can be found by the following fixed point equation, starting with $R^{ub}_{k} = C_k$:
\begin{equation}
    \label{eq:rta-fixed-point-iter}
    R^{ub}_{k} \leftarrow C_k + \left \lfloor \frac{1}{n_c} \sum_{i \in hp(k)} \min(\hat{W}_i(R^{ub}_{k}), R^{ub}_{k} - C_k + 1) \right \rfloor
\end{equation}

\section{Global Scheduling for Weakly-hard Tasks}
\label{sec:contribution-algorithm}

In this section, we present a new job-class-level algorithm for global scheduling of weakly-hard tasks.
We start by defining a deadline sequence that satisfies the weakly-hard constraint.
Our algorithm works on enforcing the defined deadline sequence to guarantee the schedulability by assigning various priorities to released jobs.
Then, we show how priorities are assigned to tasks and to released jobs.

In the next section, we show how the enforced deadline sequence facilitates the schedulability analysis.

\subsection{Critical-sequence}
\label{subsec:sequence-of-deadlines}
The $\mK$ constraint does not specify the distribution of the $m_i$ deadline misses, e.g.\ if they could happen consecutively or not.
Hence, there are different deadline sequences that satisfy the weakly-hard constraint.
We are interested in one sequence that we can enforce in our scheduling algorithm such that we guarantee the satisfiability of $\mK$.
For that end, we define $w_i$ and $h_i$.
\begin{definition}[$w_i$]\label{def:w-deadline-misses}
    It represents the maximum number of consecutive deadline misses to uniformly distribute $m_i$ in a window of $K_i$ and it is calculated as follows:
    \begin{equation}
        w_i = \max \left( \left \lfloor \frac{m_i}{K_i - m_i}\right \rfloor, 1 \right)
    \end{equation}
\end{definition}
\begin{definition}[$h_i$]\label{def:h-deadline-hits}
    It represents the number of deadline hits required per deadline miss and it is calculated as follows:
    \begin{equation}
        h_i = \left \lceil \frac{K_i - m_i}{m_i}\right \rceil
    \end{equation}
\end{definition}

$w_i,h_i$ take particular values when we consider low-tolerance or high-tolerance tasks.
\begin{lemma}
    If $m_i / K_i < 0.5$, then $w_i = 1$.
    If $m_i / K_i \geq 0.5$, then $h_i = 1$.
    \label{lem:w-h-particular-values}
\end{lemma}
\begin{proof}
	$m_i /K_i < 0.5 \Rightarrow \lfloor \frac{m_i}{K_i - m_i} \rfloor = 0$, hence, $w_i=1$. 
	Similarly, $m_i / K_i \geq 0.5 \Rightarrow  \lceil \frac{K_i - m_i}{m_i} \rceil = 1$, hence, $h_i=1$.
\end{proof}

\begin{definition}[Critical sequence]
	It is the sequence made up of {\it $h_i$} consecutive deadline hits followed by {\it $w_i$} consecutive deadline misses.
	\label{def:critical-sequence}
\end{definition}
\figurename~\ref{fig:weakly-hard-jobs} shows two critical-sequence examples, one for high-tolerance tasks and the other for low-tolerance tasks.

Our scheduling algorithm assigns a higher priority to the $h_i$ consecutive jobs, i.e. to the jobs which require to meet their deadline according to the critical-sequence.
Therefore, it is vital to prove that the critical-sequence satisfies the $\mK$.
However, the critical-sequence satisfies the  $ \overline{ \big<\begin{smallmatrix} w_i \\ w_i+h_i \end{smallmatrix}\big>}$ constraint by definition.
We show now that the critical-sequence also satisfies the constraint $\wh$.
\begin{lemma}
	For low-tolerance and high-tolerance tasks, for which $w_i$ and $h_i$ are defined as in Definition~\ref{def:w-deadline-misses} and Definition~\ref{def:h-deadline-hits} respectively, the  critical-sequence satisfies the constraint $\wh$.
	\label{lem:consecutive-constraint}
\end{lemma}
\begin{proof}
	For low-tolerance tasks $w_i=1$, hence, the following holds: 
	$\overline{\big<\begin{smallmatrix}1\\ 1 + h_i\end{smallmatrix}\big>} \equiv \overline{\big(\begin{smallmatrix}1\\ 1 + h_i\end{smallmatrix}\big)}$.
	From~\cite{bernat2001weakly}, we have $\wh \equiv \hw$. 
    Therefore, for high-tolerance tasks $h_i=1$, hence, the following holds: $\big<\begin{smallmatrix}1\\ 1 + w_i\end{smallmatrix}\big> \equiv \big(\begin{smallmatrix}1\\ 1 + w_i\end{smallmatrix}\big)$.
\end{proof}
Our goal is to prove that the critical-sequence satisfies $\mK$ and not only $\wh$.
Therefore, we have to prove that $\wh$ is harder than $\mK$.
\begin{definition}[\cite{bernat2001weakly}]
    Given two constraints, $\lambda$ and $\gamma$, we say that $\lambda$ is harder than $\gamma$, denoted by $\lambda \preccurlyeq \gamma$, if the deadline sequences that satisfy $\lambda$ also satisfy $\gamma$.
\end{definition}

\begin{lemma}
    The weakly-hard constraint $\wh$ is harder than the constraint $\mK$, formally, $\wh \preccurlyeq \mK$.
    \label{lem:harder-constraint}
\end{lemma}
\begin{proof}
    Theorem 5 of~\cite{bernat2001weakly} provides the condition that must be satisfied for one weakly-hard constraint to be harder than another.
    In this proof, we show that $\wh$ satisfies the condition to be harder than $\mK$.
     The detailed proof is in the appendix.
\end{proof}

\begin{theorem}\label{th:wh mK}
    If $\tau_i$ fulfills the constraint $\wh$, it also fulfills the constraint $\mK$.
    \label{the:new-constraint}
\end{theorem}
\begin{proof}
    This is proven by Lemma~\ref{lem:harder-constraint}, as $\wh \preccurlyeq \mK$.
\end{proof}

\begin{figure}
    \centering
        \resizebox{1\columnwidth}{!}{
            \includegraphics{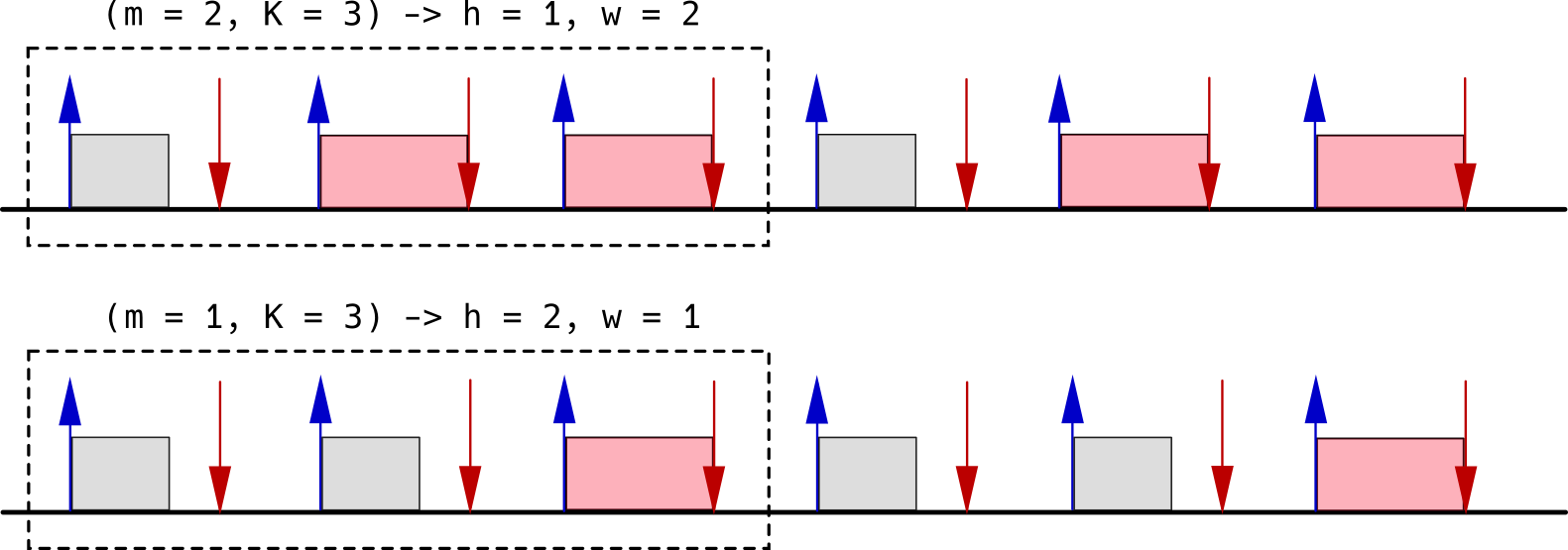}
        }
    \caption{Critical-sequence examples for a high-tolerance task (above) and a low-tolerance task (below).
    Gray boxes refer to execution finishing before the deadline, i.e., deadline hit.
    Boxes in pink refer to deadline miss.}
    \label{fig:weakly-hard-jobs}
\end{figure}

\subsection{Priority Assignment to Job-classes}
We assign a group of priorities to every weakly-hard task.
Each job of a task is assigned only one priority from this group, i.e., job-level fixed priority.
Jobs which require to meet their deadlines based on the critical-sequence will receive the highest priority of the task.
The other jobs will receive a lower priority.
In this way, a task can reduce its priority after achieving the minimum number of deadline hits ($h_i$) relaxing its interference to other tasks.

The group of priorities of a task is represented by the concept of job-classes coming from~\cite{choi2021toward}.
Each task has job-classes and every job-class has a different designated priority.
\begin{definition}\label{def:job-class}
    A task $\tau_i$ has $\mathcal{JC}_i = K_i - m_i + 1$ job-classes where every job-class, denoted by $\mathcal{JC}^{q}_{i}$ and $q$ can take values from the range $[0, K_i - m_i]$.
    Job-classes with lower values of $q$ are assigned with higher priorities, i.e. $\mathcal{JC}^{q = 0}_{i}$ and $\mathcal{JC}^{q = K_i - m_i}_{i}$ have the highest and lowest priority of the task, respectively.
\end{definition}

Table~\ref{tab:example-tasks} shows an example of three different tasks and their corresponding $q$ range.

\begin{table}[h!]
    \centering
    \caption{Example of three tasks and their $q$ ranges}
    \begin{tabular}{ c | c }
        \hline
        \rule{0pt}{3ex}
        Tasks $(C_i, D_i, T_i, \mK)$ & $q$ range \\
        \hline
        \rule{0pt}{3ex}
        $\tau_1 = (2, 6, 6, \overline{ \big(\begin{smallmatrix} 2 \\ 5 \end{smallmatrix}\big)})$ & $[0, 3]$ \\
         \rowcolor{gray!10}\rule{0pt}{3ex}
        $\tau_2 = (3, 7, 7, \overline{ \big(\begin{smallmatrix} 1 \\ 3 \end{smallmatrix}\big)})$ & $[0, 2]$ \\
        \rule{0pt}{3ex}
        $\tau_3 = (2, 8, 8, \overline{ \big(\begin{smallmatrix} 2 \\ 3 \end{smallmatrix}\big)})$ & $[0, 1]$ \\
        \hline
    \end{tabular}
    \label{tab:example-tasks}
\end{table}

Every job-class has a different priority, i.e. the same priority is not shared between job-classes of different tasks.
Algorithm~\ref{alg:priority_assignment} shows how priorities are assigned to each job-class.
First, tasks are sorted in ascending order of deadline (Line 2).
In case of two or more tasks have the same deadline, the one with lower $m_i$ is ordered first.
If tasks have also the same $m_i$, the order between them is selected randomly.
Then, the total number of priorities is calculated by counting number of job-classes between all tasks (Line 5).
Finally, the priority is assigned to each job-class level by iterating over them (from Line 8 until Line 12).
\figurename~\ref{fig:priority-example-tasks} shows the job-classes of tasks from Table~\ref{tab:example-tasks} after running the priority assignment algorithm.

\begin{algorithm}[t!]
    \small
    \DontPrintSemicolon

    {\bf Input:} taskset $\mathcal{T}$\;
    $sort\_tasks\_ascending\_deadline(\mathcal{T})$\;

    \For{$\tau_i \in \mathcal{T}$}
    {
        $\mathcal{JC}_i \gets  K_i - m_i + 1$\;
    }
    $\mathcal{JC} \gets \sum_{\forall \tau_i \in \mathcal{T}}\mathcal{JC}_i$\;
    $prio \gets \mathcal{JC}$\;
    $\mathcal{JC}^{max} \gets max\{\mathcal{JC}_i | \forall \tau_i \in \mathcal{T}\}$\;
    \For { $q \gets0; q<\mathcal{JC}^{max}; q \gets q + 1$}
    {
        \For {$\tau_i \in \mathcal{T}$}
        {
            \If {$q < \mathcal{JC}_i$}
            {
                $ \mathcal{JC}_i^q \gets prio$\;
                $prio \gets prio - 1$
            }
        }
    }

    \caption{Priority assignment to job-classes.}
    \label{alg:priority_assignment}
\end{algorithm}

\begin{figure}[t!]
    \centering
        \resizebox{0.4\columnwidth}{!}{
            \includegraphics{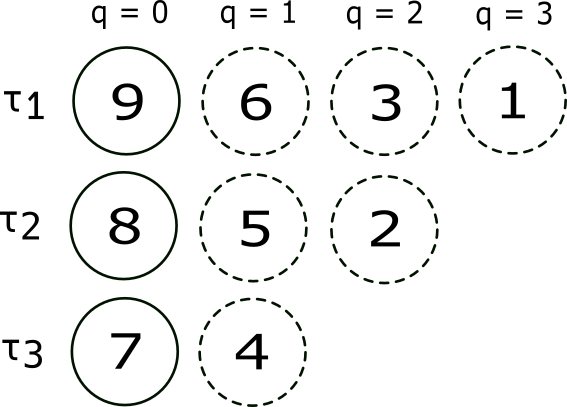}
        }
    \caption{Priority assignment for tasks in Table~\ref{tab:example-tasks}.
    Solid circles represent the highest priority for each task, thus, each task has only one solid circle.
    Note that jobs of $\tau_1$ may get the highest priority among all possible priorities (9) and the lowest possible priority (1).}
    \label{fig:priority-example-tasks}
\end{figure}

\subsection{Scheduling Routine: Assigning Priorities to Released Jobs}
Every time a job is released, the scheduler assigns it to a job-class based on the previous deadline misses/hits.
For selecting to which particular job-class a job should be assigned, we define job-level:
\begin{definition}
    The job-level $jl_i$ is a variable of a task $\tau_i$ which is used to select the job-class $\mathcal{JC}^{q}_{i}$ of the released job.
    The value of $q$ is selected according to $q = max(0, jl_i)$.
    The initial value of $jl_i$ is $-(h_i - 1)$ and every time a job meets its deadline, $jl_i$ is increased by one until $K_i - m_i$.
    When $w_i$ deadline misses happened, the value of $jl_i$ is restored to $-(h_i - 1)$.
    \label{def:job-level}
\end{definition}

Based on how $jl_i$ is updated, the following consequences can be deduced.
For high-tolerance tasks, the starting value of $jl_i$ is zero, since for that kind of tasks $h_i$ is one;
and for low-tolerance tasks, every time a deadline is missed, $jl_i$ is restored to $-(h_i - 1)$, since for those kind of tasks $w_i$ is one (see Lemma~\ref{lem:w-h-particular-values}).

\figurename~\ref{fig:scheduling-algo} shows the transitions between job-classes for low-tolerance and high-tolerance tasks.

\begin{figure}
    \centering
        \resizebox{0.6\columnwidth}{!}{
            \includegraphics{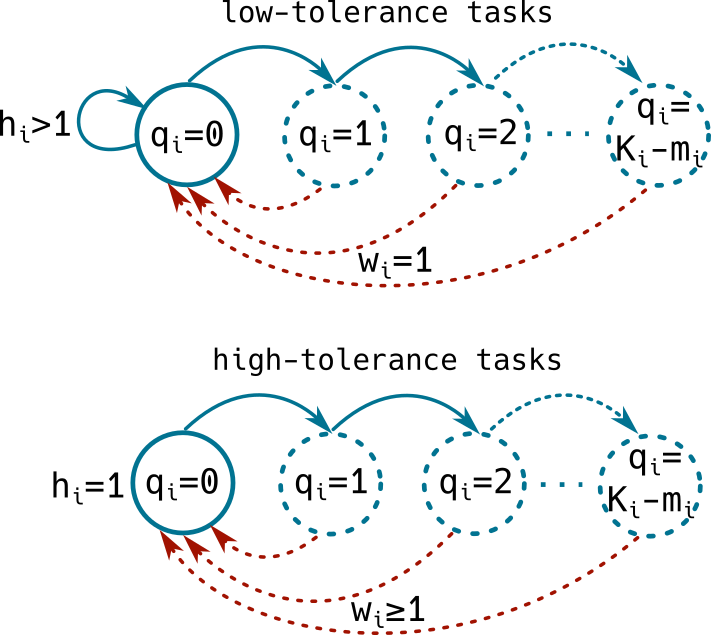}
        }
    \caption{Job-classes transitions for low-tolerance and high-tolerance tasks.
    Solid circles represent the highest priority, hence, jobs which are assigned to the priority represented by the solid circle are guaranteed to meet their deadlines.}
    \label{fig:scheduling-algo}
\end{figure}

\section{Response Time Analysis Extension for Weakly-hard Real-time Tasks}
\label{sec:contribution-analysis}
This section presents a schedulability analysis for the proposed job-class-level scheduling algorithm described in the previous section.
The analysis is an extension of RTA for weakly-hard real-time tasks scheduled by our algorithm.
As mentioned in Section~\ref{sec:RTA}, RTA provides a sufficient condition for schedulability by proving that the response time of a task is not longer than its deadline.
The response time of a task is prolonged as much as tasks with higher priorities interfere with the execution of the analyzed task.
The extension of RTA presented here bounds the interference over lower priority tasks by taking into account only jobs belonging to $\mathcal{JC}^{q = 0}_{i}$.
Remember that jobs in $\mathcal{JC}^{q = 0}_{i}$ have the highest assignable priority to their tasks and are selected based on the minimum consecutive deadline hits ($h_i$) from the critical-sequence.
Moreover, considering only jobs in $\mathcal{JC}^{q = 0}_{i}$ reduces the interference to lower priority jobs.

\begin{lemma}\label{lem:wh sufficint condition}
	For a task $\tau_i$, to meet its constraint $\wh$, it is sufficient that the jobs belonging to $\mathcal{JC}^{q = 0}_{i}$ meet their deadlines.
\end{lemma}
\begin{proof}
	Our proposed scheduling assigns $h_i$ jobs to the $\mathcal{JC}^{q = 0}_{i}$ every time $\tau_i$ misses $w_i$ deadlines as \figurename~\ref{fig:scheduling-algo} illustrates.
	Hence, if all jobs belonging to $\mathcal{JC}^{q = 0}_{i}$ meet their deadlines, $\tau_i$ meets its constraint $\wh$ regardless whether the other jobs, which belong to other job-classes, meet their deadlines or not.
\end{proof}
Consequently, this section considers only the jobs in job-classes $\mathcal{JC}^{q = 0}_{i}$. This section starts by proving that it is possible to extend RTA for the proposed scheduling algorithm.
Finally, the schedulability condition for a task set is defined.

\subsection{Proving RTA Extension}
We prove the extension of RTA for our algorithm by showing how to bound the interference over jobs in job-classes $\mathcal{JC}^{q = 0}_{i}$.
First, we show that only the interference between the highest priority job-classes have to be considered.
Then, we evaluate the interference based on the critical-sequence.
This is done separately for low-tolerance and high-tolerance tasks because of the differences in the critical-sequences (see Definition~\ref{def:critical-sequence}).

With the following lemma, we show that only the response time of jobs belonging to $\mathcal{JC}^{q = 0}_{i}$ are of interest.
\begin{lemma}
    A job of a task $\tau_k$ in $\mathcal{JC}^{q = 0}_{k}$ suffers interference only from the jobs of a task $\tau_i$ in $\mathcal{JC}^{q = 0}_{i}$, if the priority of $\mathcal{JC}^{q = 0}_{i}$ is higher than the priority of $\mathcal{JC}^{q = 0}_{k}$.
\end{lemma}
\begin{proof}
    The Algorithm~\ref{alg:priority_assignment} assigns a priority value to job-classes starting by $q = 0$ and every time a priority is assigned, the next priority value is reduced by one.
    In this way, priority values of job-classes $\mathcal{JC}^{q \geq 1}$ are always lower than the ones assigned to job-classes $\mathcal{JC}^{q = 0}$.
    From which it follows that jobs of a task $\tau_k$ which belong to job-class $\mathcal{JC}^{q = 0}_k$ suffer interference of other jobs in $\mathcal{JC}^{q = 0}_i$, only if the priority of $\mathcal{JC}^{q = 0}_i$ is higher than the priority of $\mathcal{JC}^{q = 0}_k$.
\end{proof}

To bound the interference induced by $\tau_i$ on $\tau_k$, 
we must bound the maximum number of jobs in $\mathcal{JC}^{q = 0}_i$.
In the worst-case, jobs of $\tau_i$ are assigned to $\mathcal{JC}^{q = 0}_i$ $h_i$ times every $w_i+h_i$ jobs.
For  a high-tolerance task $\tau_i$, the jobs in $\mathcal{JC}^{q = 0}_i$ are $w_i$ apart.
This allows to consider the workload coming from such a task as the same workload produced by a task with a longer inter-arrival time that is equal to $(w_i + 1) T_i$.
Formally writing this:
\begin{lemma}
    \label{lem:workload-equivalent}
    The workload of a high-tolerance task $\tau_i$ with constraint $\mK$ is calculated as it were coming from an equivalent hard real-time task $\tau^{eq}_i = (C_i, D_i, (w_i + 1) T_i)$.
\end{lemma}
\begin{proof}
	The jobs of the hard real-time task $\tau^{eq}_i = (C_i, D_i, (w_i + 1) T_i)$ have the same worst-case execution time $C_i$ as the jobs of $\tau_i$ in $\mathcal{JC}^{q = 0}_i$, and occur at the same inter-arrival time $(w_i+1)T_i$.
    Therefore they induce the same interference.
\end{proof}

For low-tolerance tasks, the minimum inter-arrival time between two jobs belonging to $\mathcal{JC}^{q = 0}_{i}$ is $T_i$.
Therefore, bounding the interference of low-tolerance tasks requires considering $T_i$ as the inter-arrival time and excluding the jobs that do not belong to $\mathcal{JC}^{q = 0}$ from the workload bound.
Hence, we update the workload bound defined in Equation~\eqref{eq:rta-bounded-workload} as follows.
In the first term, the update subtracts the number of jobs with lower priority from $N_i$ over the interval of interest $L$.
For updating the workload coming from the carry-out job, the second term of $\hat{W}_i(L)$ is made equal to zero when a carry-out job has a lower priority than $\mathcal{JC}^{q = 0}_{i}$.

For updating the first term, the number of lower priority jobs is defined by the following lemma:
\begin{lemma}
    \label{lem:carry-in-and-body-jobs}
    The minimum number of jobs of a low-tolerance task $\tau_i$ that have lower priorities than $\mathcal{JC}^{q = 0}_{i}$ during the interval $L$ is given by:
\begin{equation}
    \label{eq:rta-Oi}
    O_i(L) = \left \lfloor \frac{L + R^{ub}_{i} - C_i - s_i}{T_i (h_i + 1)} \right \rfloor
\end{equation}
\end{lemma}
\begin{proof}
    The critical-sequence for a low-tolerance task allows one deadline miss after $h_i$ deadline hits.
    This means that the maximum distance between two jobs of $\tau_i$ with a priority less than $\mathcal{JC}^{q = 0}_{i}$ is $T_i (h_i + 1)$.
\end{proof}

For removing the contribution of carry-out jobs, we introduce a variable that enables the second term in $\hat{W}_i(L)$ only if the carry-out jobs are in $\mathcal{JC}^{q = 0}_{i}$.
\begin{definition}
    \label{lem:carry-out-job}
    The variable $a_i$ equals one by default and zero when carry-out jobs are not in $\mathcal{JC}^{q = 0}_{i}$.
\begin{equation}
    \label{eq:rta-ai}
    a_i = 1 - \left \lfloor \frac{N_i(L) \mod(h_i + 1)}{h_i} \right \rfloor
\end{equation}
\end{definition}
    The value of $N_i \mod(h_i + 1)$ gives the number of jobs of the current critical-sequence.
    When this value is $h_i$, the next carry-out job belongs to a lower priority job-class.
    Furthermore, the result of the floor is one when: $N_i \mod(h_i + 1) = h_i$.
    It follows that the value of $a_i$ is zero when the carry-out job belongs to a different job-class than $\mathcal{JC}^{q = 0}_{i}$.

Then, we update the workload bound equation as follows:
\begin{lemma}
    \label{lem:workload-lower-tolerance-tasks}
    The workload bound function of a low-tolerance task $\tau_i$ is updated for removing non-interfering jobs as follow:
\begin{multline}
    \label{eq:rta-Workload}
    \hat{W}_i(L) = (N_i(L) - O_i(L)) C_i \\
    + a_i . \min(C_i, (L - R^{ub}_i - C_i - s_i) \mod T_i)
\end{multline}
\end{lemma}
\begin{proof}
    The proof is given by Lemma~\ref{lem:carry-in-and-body-jobs} and Definition~\ref{lem:carry-out-job}.
\end{proof}

Algorithm~\ref{alg:compute_workload} shows the function for calculating $\hat{W}_i$ which is used in the fixed point iteration of Equation~\eqref{eq:rta-fixed-point-iter}.

\begin{algorithm}[t!]
    \small
	\DontPrintSemicolon

    {\bf Input:} Length L, task $\tau_i$\;
    {\bf Output:} Workload  $\hat{W}_i$\;    

    $s_i \gets D_i - R^{ub}_i$ {\textcolor{black!10!blue}{// Slack}}\;
    \eIf{$m_i / K_i >= 0.5$} 
    {
        {\textcolor{black!10!blue}{// High-tolerance task}}\;
        $T_i \gets (w_i+1) T_i$\;
        $O_i \gets 0$\;
        $a_i \gets 1$\;
    }
    {
        {\textcolor{black!10!blue}{// Low-tolerance task}}\;
        $O_i \gets \left \lfloor \frac{L + R^{ub}_{i} - C_i - s_i}{T_i (h_i + 1)} \right \rfloor$ {\textcolor{black!10!blue}{// Equation~\eqref{eq:rta-Oi}}}\;
        $a_i \gets 1 - \left \lfloor \frac{N_i \mod(h_i + 1)}{h_i} \right \rfloor$ {\textcolor{black!10!blue}{// Equation~\eqref{eq:rta-ai}}}\;
    }
    $N_i \gets  \left \lfloor \frac{L + R^{ub}_i - C_i - s_i}{T_i} \right \rfloor$  {\textcolor{black!10!blue}{// Equation~\eqref{eq:rta-number-jobs}}}\;
    $\hat{W}_i \gets (N_i - O_i) C_i + a_i . \min(C_i, (L - R^{ub}_i - C_i - s_i) \mod T_i) $ {\textcolor{black!10!blue}{// Equation~\eqref{eq:rta-Workload}}}\;
    \Return $\hat{W}_i$\;

	\caption{Workload computation.}
    \label{alg:compute_workload}
\end{algorithm}

\begin{lemma}
    \label{lem:rta-extension}
    $R^{ub}_{k}$ as given in~\eqref{eq:rta-fixed-point-iter} is a safe upper bound on the response time of jobs in $\mathcal{JC}^{q = 0}_{k}$, where $\hat{W}_i$ is bounded as in~\eqref{eq:rta-Workload}.
\end{lemma}
\begin{proof}
    From Algorithm~\ref{alg:compute_workload}, the workload contribution of high-tolerance tasks uses the approach proven in Lemma~\ref{lem:workload-equivalent} and for low-tolerance tasks, it uses the approach proven in Lemma~\ref{lem:workload-lower-tolerance-tasks}.
\end{proof}

\subsection{Schedulability Analysis}
We introduce the following theorem for schedulability of a weakly-hard real-time task $\tau_i$ scheduled by our algorithm:
\begin{theorem}
    \label{the:sufficient-condition-weakly}
    For a task $\tau_i$ with the constraint $\mK$, meeting the deadlines of the jobs belonging to $\mathcal{JC}^{q = 0}_{k}$ is a sufficient condition for the schedulability of $\tau_i$ .
\end{theorem}
\begin{proof}
	It follows directly from Theorem~\ref{th:wh mK} and Lemma~\ref{lem:wh sufficint condition}.
\end{proof}

Then, the schedulability of a task set is defined as follows:
\begin{corollary}
    A task set is schedulable by our scheduling algorithm if for every task, the sufficient condition given in Theorem~\ref{the:sufficient-condition-weakly} is satisfied.
\end{corollary}

\section{Evaluation}
\label{sec:evaluation}
Our global scheduling depends on the transformation of the weakly-hard constraint $\mK$ to the $\wh$ constraint which has smaller window. 
In this section, we study first the limitations introduced by adopting a harder constraint than $\mK$.
  
We also evaluate our global scheduling for weakly-hard real-time tasks in this section.
To the best of our knowledge, there is no other global scheduling analysis for weakly-hard real-time tasks.
Furthermore, since we are interested in RTEMS and the available global schedulers are EDF and RM, we compare the proposed scheduling algorithm against the results of RTA (Section~\ref{sec:RTA}) for RM and EDF.
The experiments are based on the analysis of task sets randomly generated using the UUnifast algorithm~\cite{bini2005measuring}.
For a given total utilization, we calculate the percentage of schedulable task sets, known as schedulability ratio.
Additionally, we ran our experiments for different values of $K_i$ and measured the computation times for different number of tasks.

Furthermore, in order to analyse the scalability of our approach, we also compare our analysis (here labeled as RTA WH) against the Integer Linear Programming (ILP) approach proposed in~\cite{sun2017MILP} and the Job-Class Level (JCL) proposed in~\cite{choi2021toward}.

Finally, at the end of this section, we show the execution time distribution for assigning priorities to released jobs.
This last experiment runs on the same hardware platform used in CALLISTO.

\subsection{Transformation Cost}
For analyzing the limitation introduced by transforming $\mK$ to $\wh$, where $\wh$ is harder than $\mK$, we count the possible deadline sequences which satisfy both constraints.
Algorithm~\ref{alg:limitation-calculation} counts these possible solutions based on a given $\mK$.
It starts by calculating $h_i$ and $w_i$ for a given $\mK$.
Then, it creates two binary trees of depth $K_i$, one starting with zero and the other with one.
These binary trees are used to extract the different deadline sequences combinations which are in fact the branches of the trees, see \figurename~\ref{fig:binary-tree}.
Later, the solutions for $\mK$ and $\wh$ are counted from the branches.

\begin{algorithm}
	\small
	\DontPrintSemicolon
	
	{\bf Input:} constraint $\mK$\;
	{\bf Output:} relation between solutions for the harder and the original constraints\;
	
	$w_i \gets \max \left( \left \lfloor \frac{m_i}{K_i - m_i}\right \rfloor, 1 \right)$ {\textcolor{black!10!blue}{// Definition~\ref{def:w-deadline-misses}}}\;
	$h_i \gets \left \lceil \frac{K_i - m_i}{m_i}\right \rceil$ {\textcolor{black!10!blue}{// Definition~\ref{def:h-deadline-hits}}}\;
	
	$trees \gets create\_binary\_trees(K_i)$\;
	$branches \gets extract\_branches(trees)$\;
	$solutions \gets count\_solutions(\mK, branches)$\;
	$harder\_solutions \gets count\_solutions(\wh, branches)$\;
	\Return $harder\_solutions / solutions$\;
	
	\caption{Algorithm for counting possible solutions.}
	\label{alg:limitation-calculation}
\end{algorithm}
\begin{figure}
	\centering
	\resizebox{0.6\columnwidth}{!}{
		\includegraphics{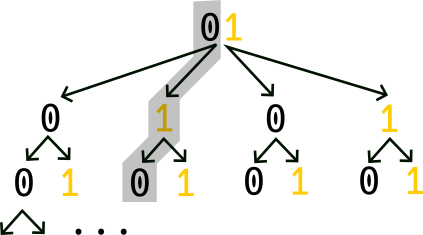}
	}
	\caption{Binary trees created for counting deadline sequences (example of sequence highlighted).}
	\label{fig:binary-tree}
\end{figure}

Results for different values of $\mK$ are shown in Table~\ref{tab:limitation-harder-constraint}.
For low-tolerance tasks, the omitted deadline sequences can be very high when $m_i/K_i$ is near to 0.5.
On the other hand, there are less omitted sequences for high-tolerance.
For both kind of tasks, the omitted sequences increases when considering a constraint which is multiple of another, e.g.\ $\mK = \wharg{8}{20}$ and $\mK = \wharg{2}{5}$.
This is explained because some deadline sequences with consecutive deadline misses are not counted.
Despite this limitation, our experiments show better results than traditional hard real-time scheduling analysis.
Furthermore, the results obtained here are theoretical since not all the uncounted deadline sequences will also mean schedulable solutions.
\begin{table}
	\centering
    \caption{Limitation for harder constraints.}
	\resizebox{1\columnwidth}{!}{
		\begin{tabular}{ c | c | c }
			\hline
            \rule{0pt}{3ex}
			$\mK$ & $\wh$ & $\wh$ solutions / $\mK$ solutions\\
			\hline
            \rule{0pt}{3ex}
            $\wharg{1}{5}$ & $\wharg{1}{5}$ & 1.0 \\
             \rowcolor{gray!10}\rule{0pt}{3ex}
            $\wharg{2}{5}$ & $\wharg{1}{3}$ & 0.5625 \\
            \rule{0pt}{3ex}
            $\wharg{3}{5}$ & $\wharg{1}{2}$ & 0.5 \\
             \rowcolor{gray!10}\rule{0pt}{3ex}
            $\wharg{4}{5}$ & $\wharg{4}{5}$ & 1.0 \\
             \rule{0pt}{3ex}
            $\wharg{4}{10}$ & $\wharg{1}{3}$ & 0.1554 \\
            \rowcolor{gray!10}\rule{0pt}{3ex}
            $\wharg{8}{10}$ & $\wharg{4}{5}$ & 0.9003 \\
             \rule{0pt}{3ex}
            $\wharg{8}{20}$ & $\wharg{1}{3}$ & 0.01040 \\
            \rowcolor{gray!10}\rule{0pt}{3ex}
            $\wharg{16}{20}$ & $\wharg{4}{5}$ & 0.7511 \\
			\hline
		\end{tabular}
	}
    \label{tab:limitation-harder-constraint}
\end{table}

\pdfsuppresswarningpagegroup=1

\begin{figure*}[t!]
	\centering
	\resizebox{2.1\columnwidth}{!}{
		\begin{tabular}{ c  c  c }
			(a) $n_c=2$ &  (b) $n_c=4$& (c) $n_c=8$\\ 
			\includegraphics[width=1\columnwidth]{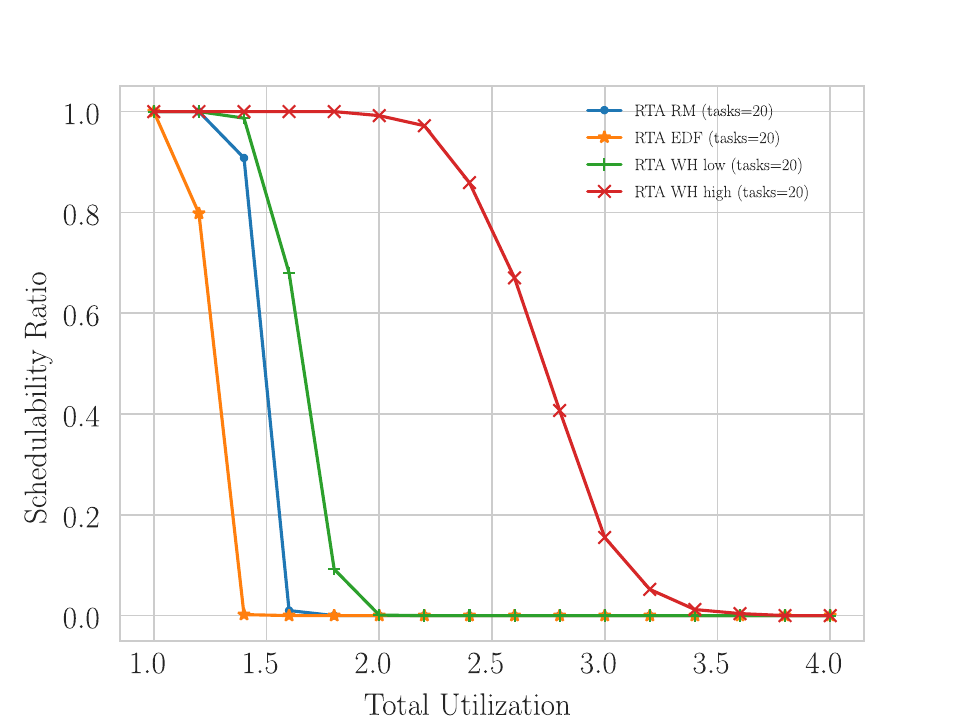} &
			\includegraphics[width=1\columnwidth]{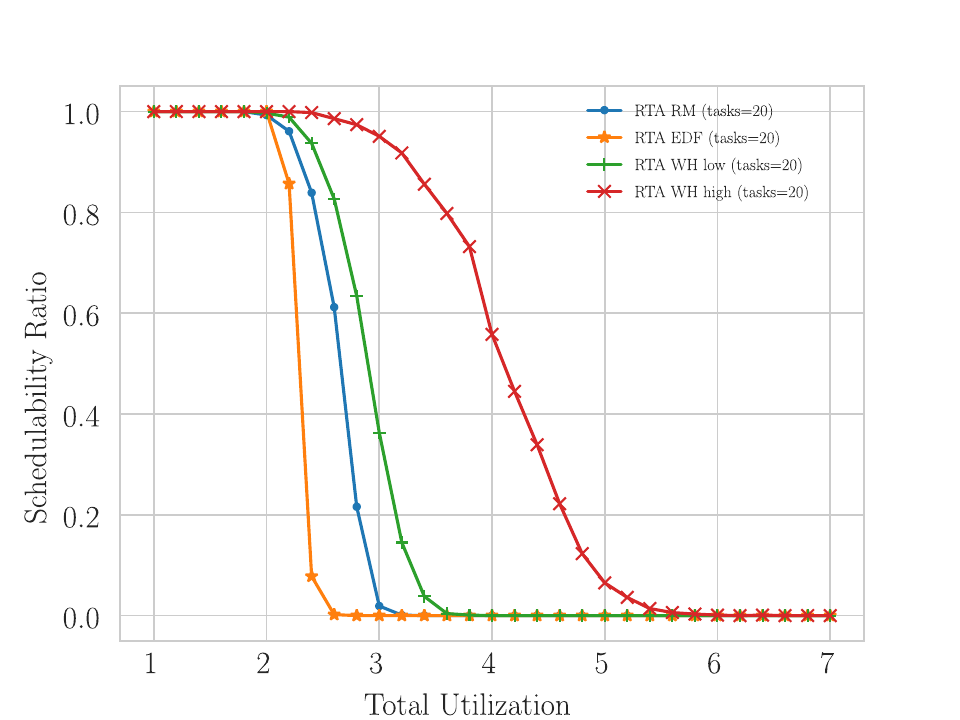} &
			\includegraphics[width=1\columnwidth]{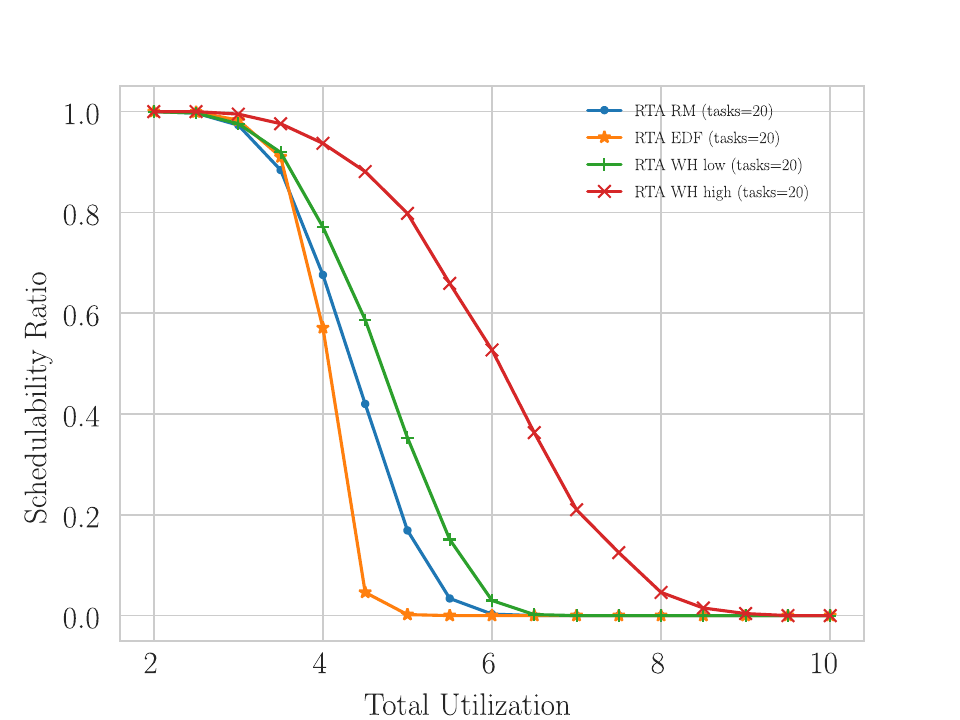} \\
		\end{tabular}
	}
	\caption{Schedulability ratio.}
	\label{fig:sched-ratio}
\end{figure*}

\begin{figure}[h!]
    \centering
    \includegraphics[width=0.67\columnwidth]{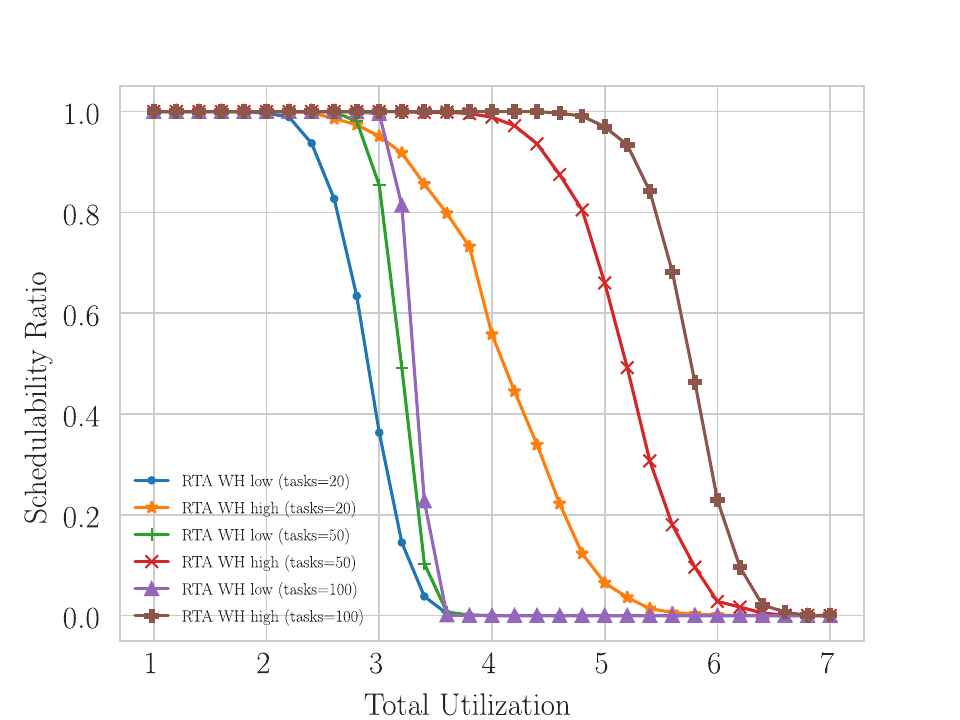}
	\caption{Schedulability ratio for 20, 50 and 100 tasks (4 cores)}
    \label{fig:sched-ratio-tasks}
\end{figure}

\begin{figure}[h!]
    \centering
    \includegraphics[width=0.67\columnwidth]{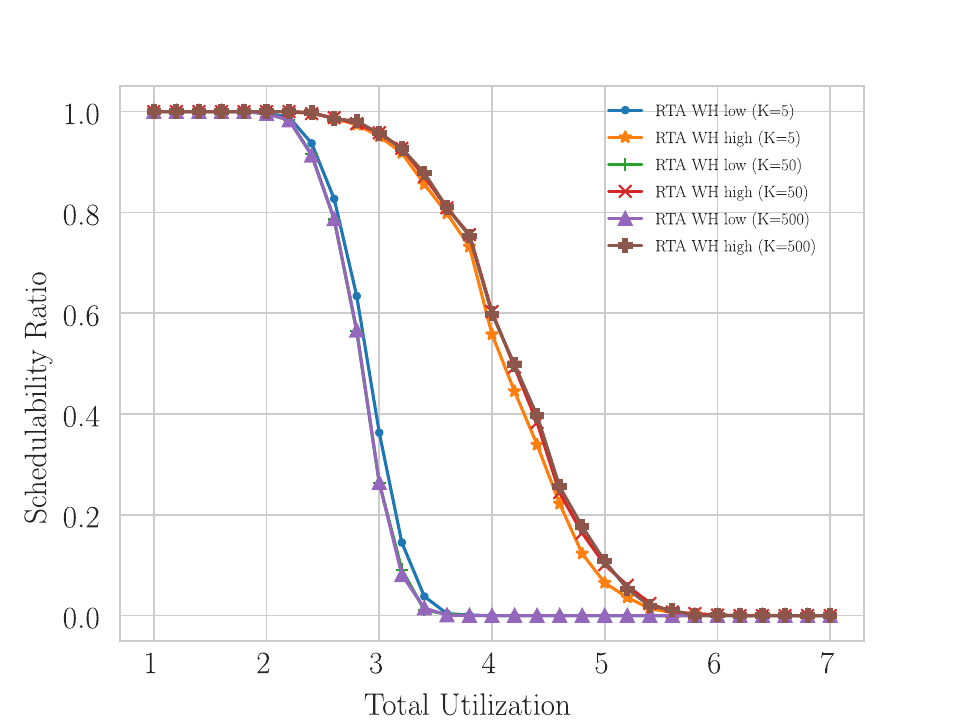}
    \caption{Schedulability ratio for $K_i$ equals to 5, 50 and 500 (4 cores)}
    \label{fig:sched-ratio-k}
\end{figure}

\begin{figure*}
    \centering
    \resizebox{2\columnwidth}{!}{
    \begin{tabular}{ c  c  c }
        \includegraphics[width=\columnwidth]{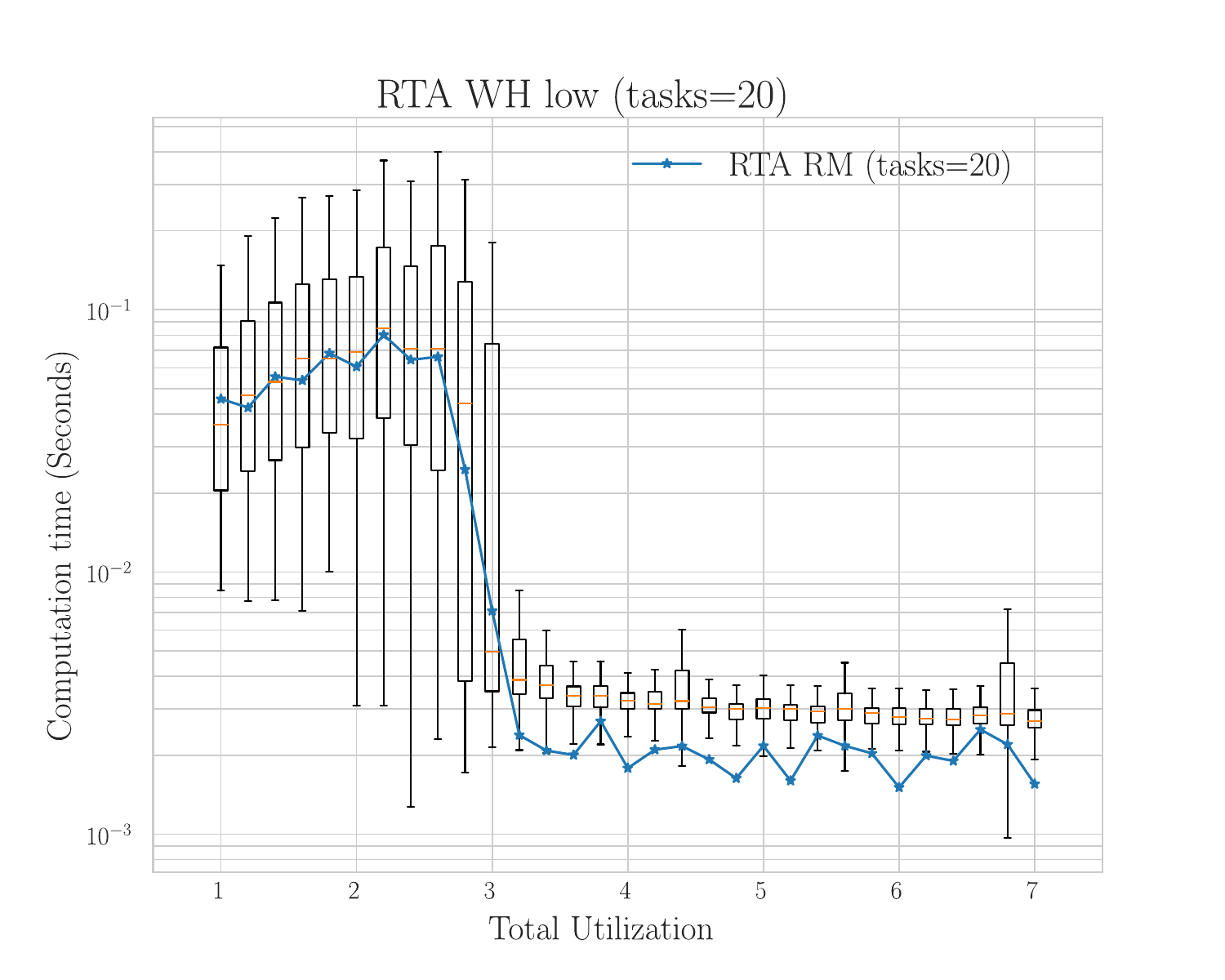} &
        \includegraphics[width=\columnwidth]{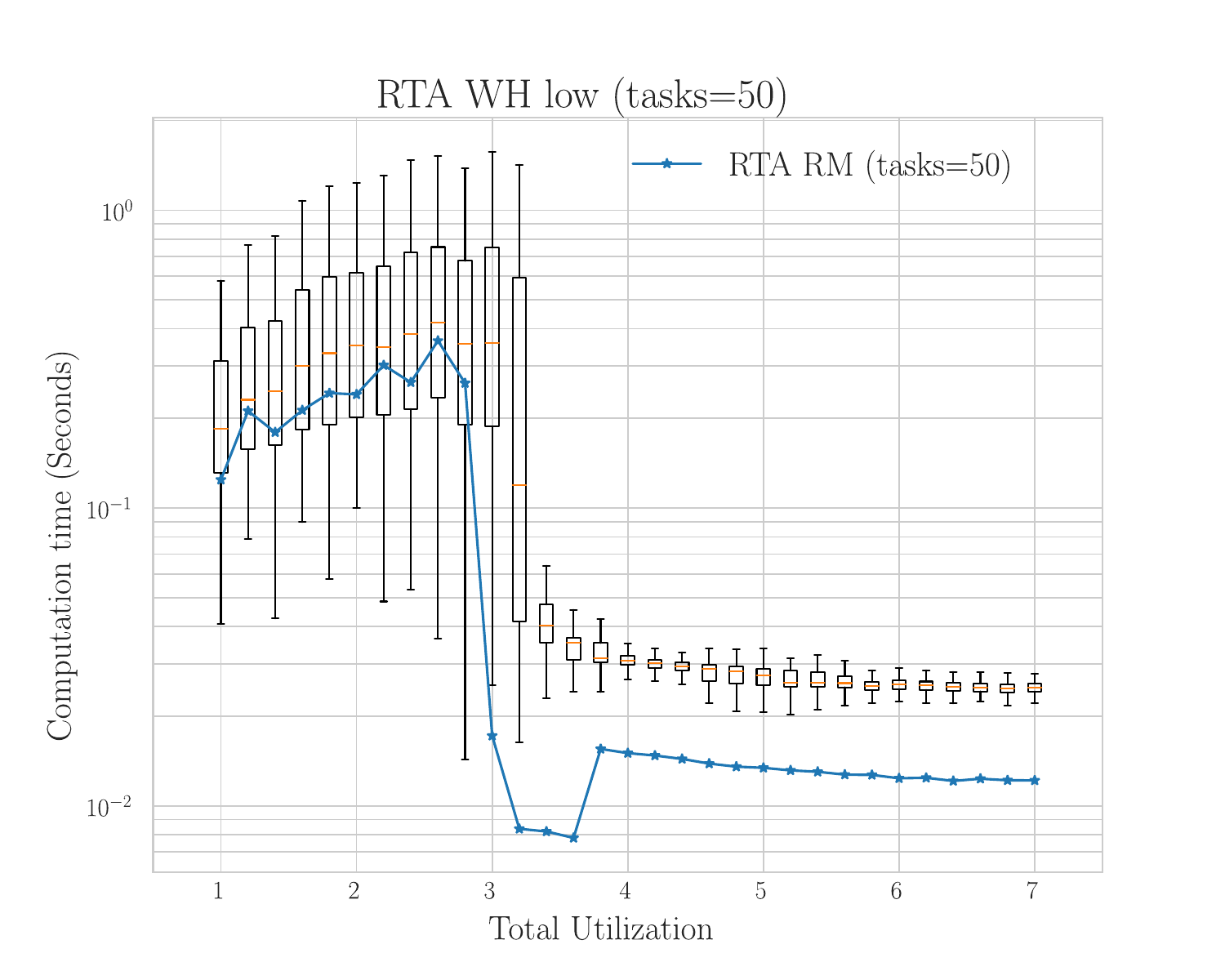} &
        \includegraphics[width=\columnwidth]{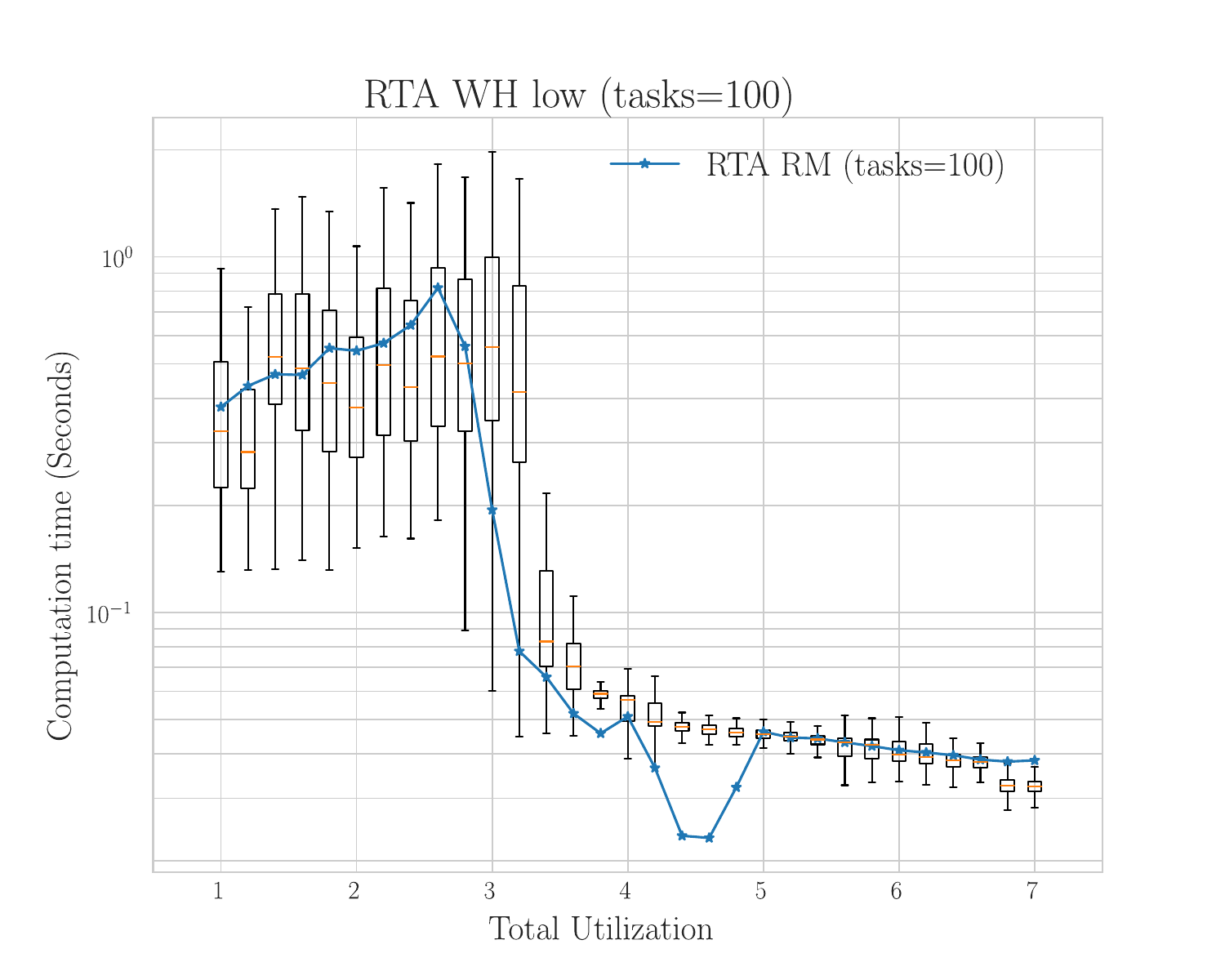} \\
        \includegraphics[width=\columnwidth]{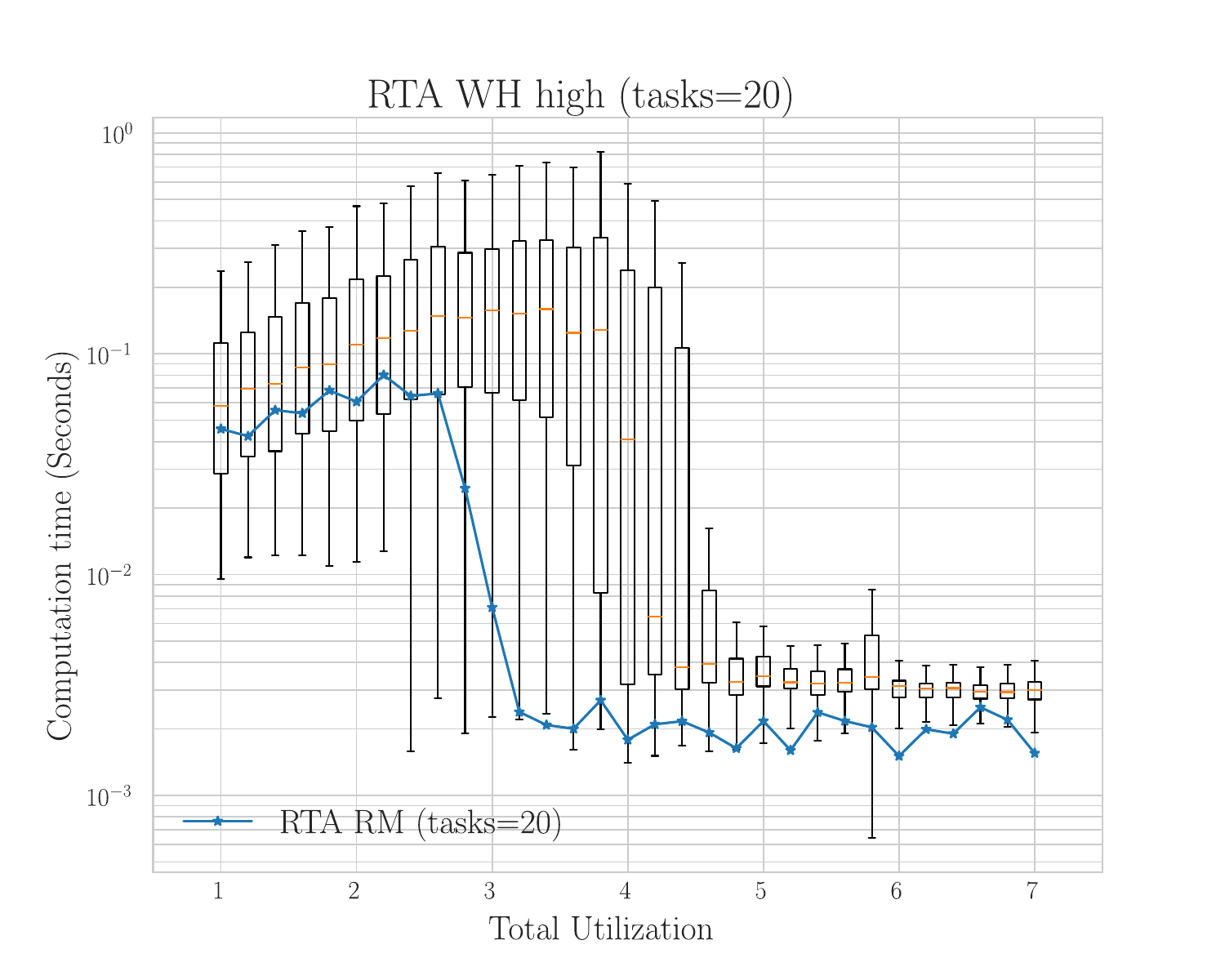} &
        \includegraphics[width=\columnwidth]{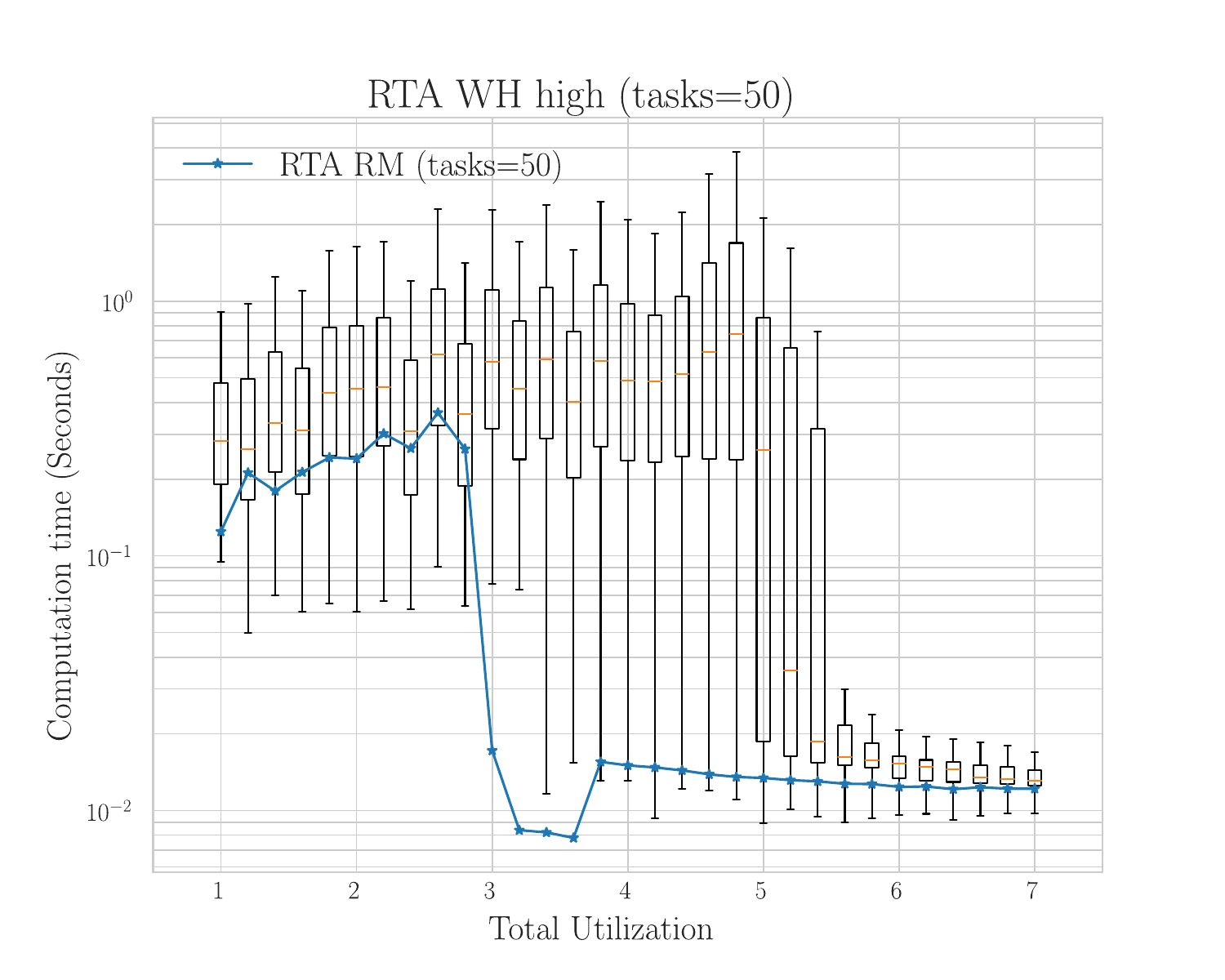} &
        \includegraphics[width=\columnwidth]{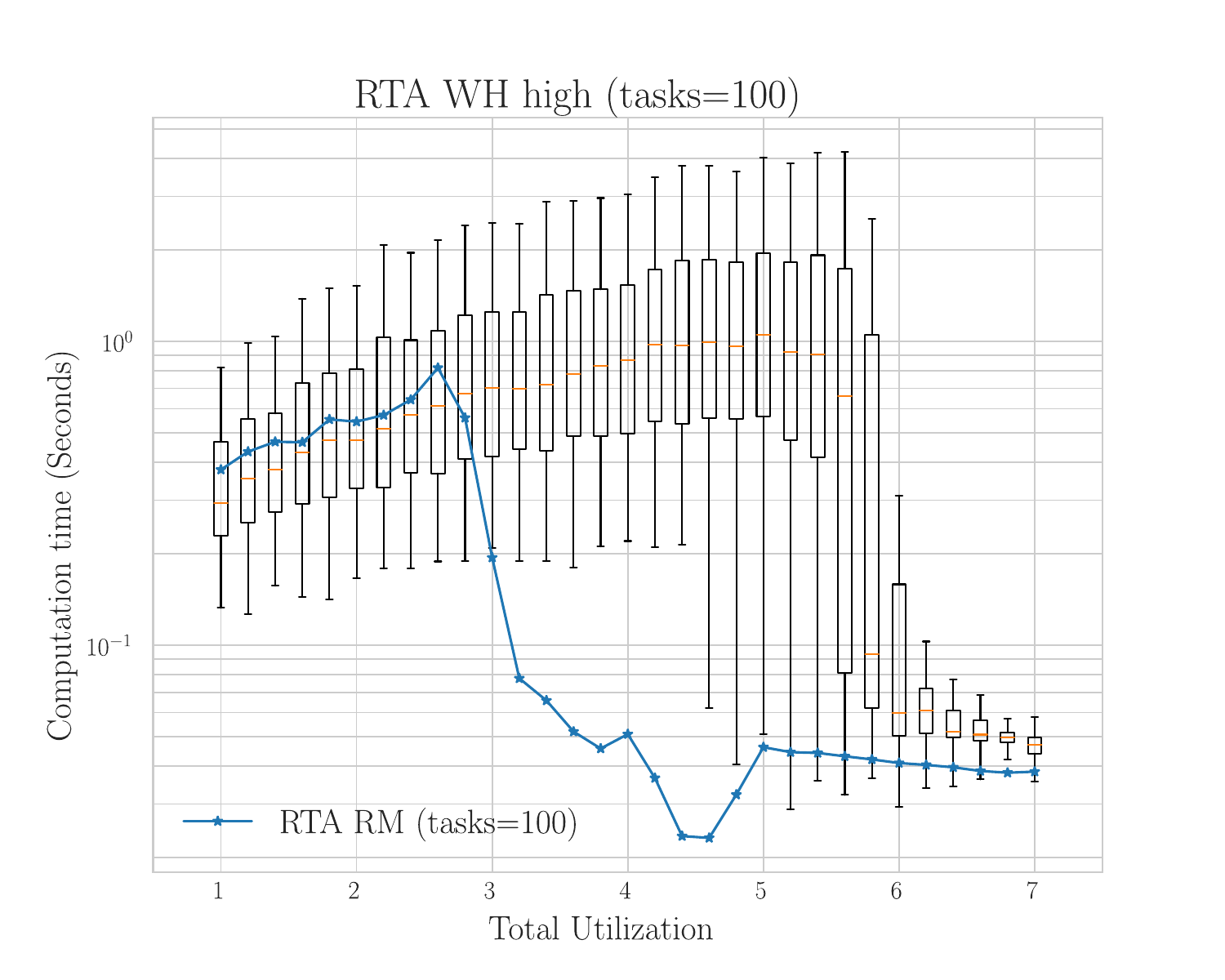} \\
    \end{tabular}
    }
    \caption{Computation time for 20, 50 and 100 tasks (4 cores).}
    \label{fig:computation-time-tasks}
\end{figure*}

\begin{figure*}
    \centering
    \resizebox{2\columnwidth}{!}{
    \begin{tabular}{ c  c  c }
        \includegraphics[width=\columnwidth]{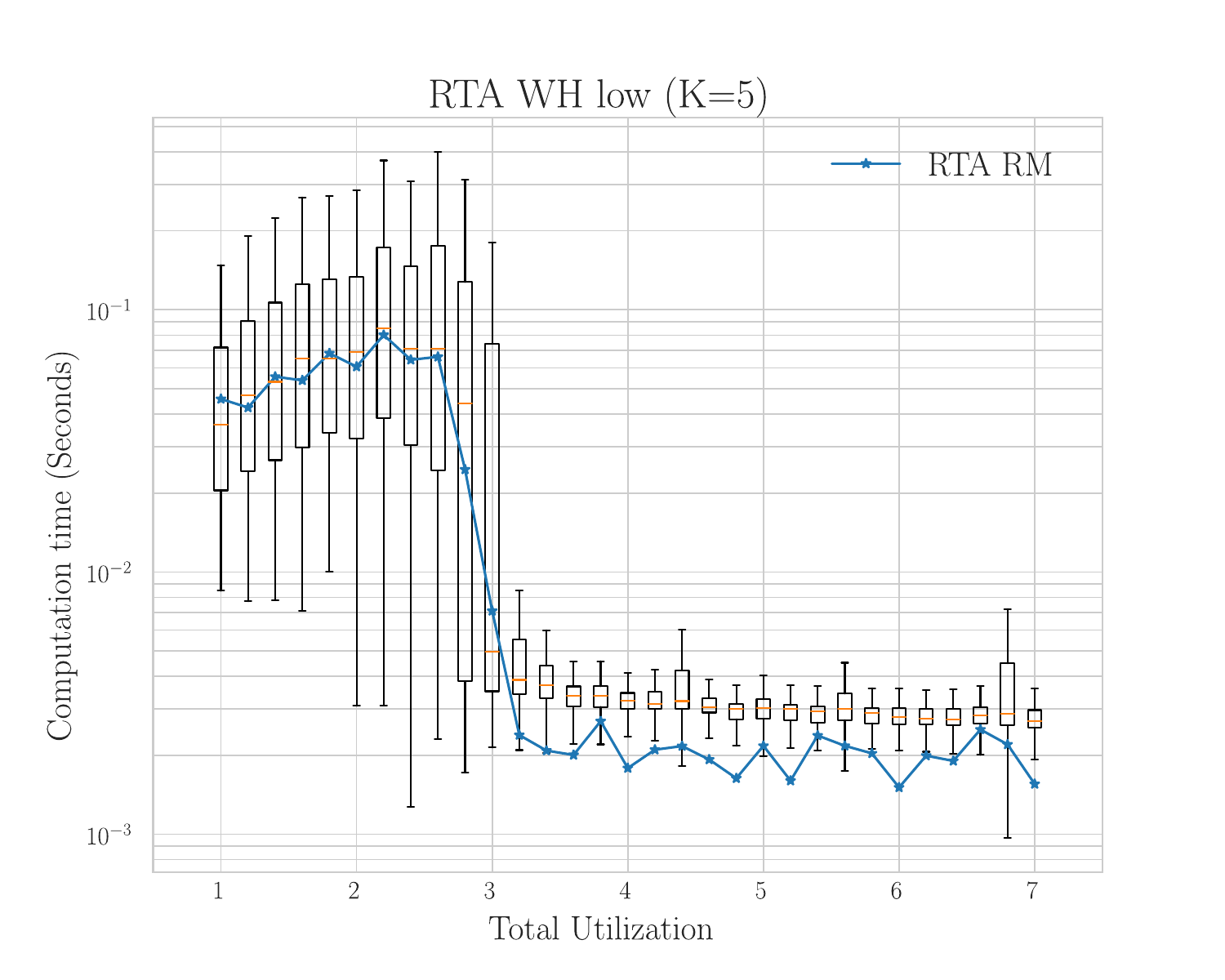} &
        \includegraphics[width=\columnwidth]{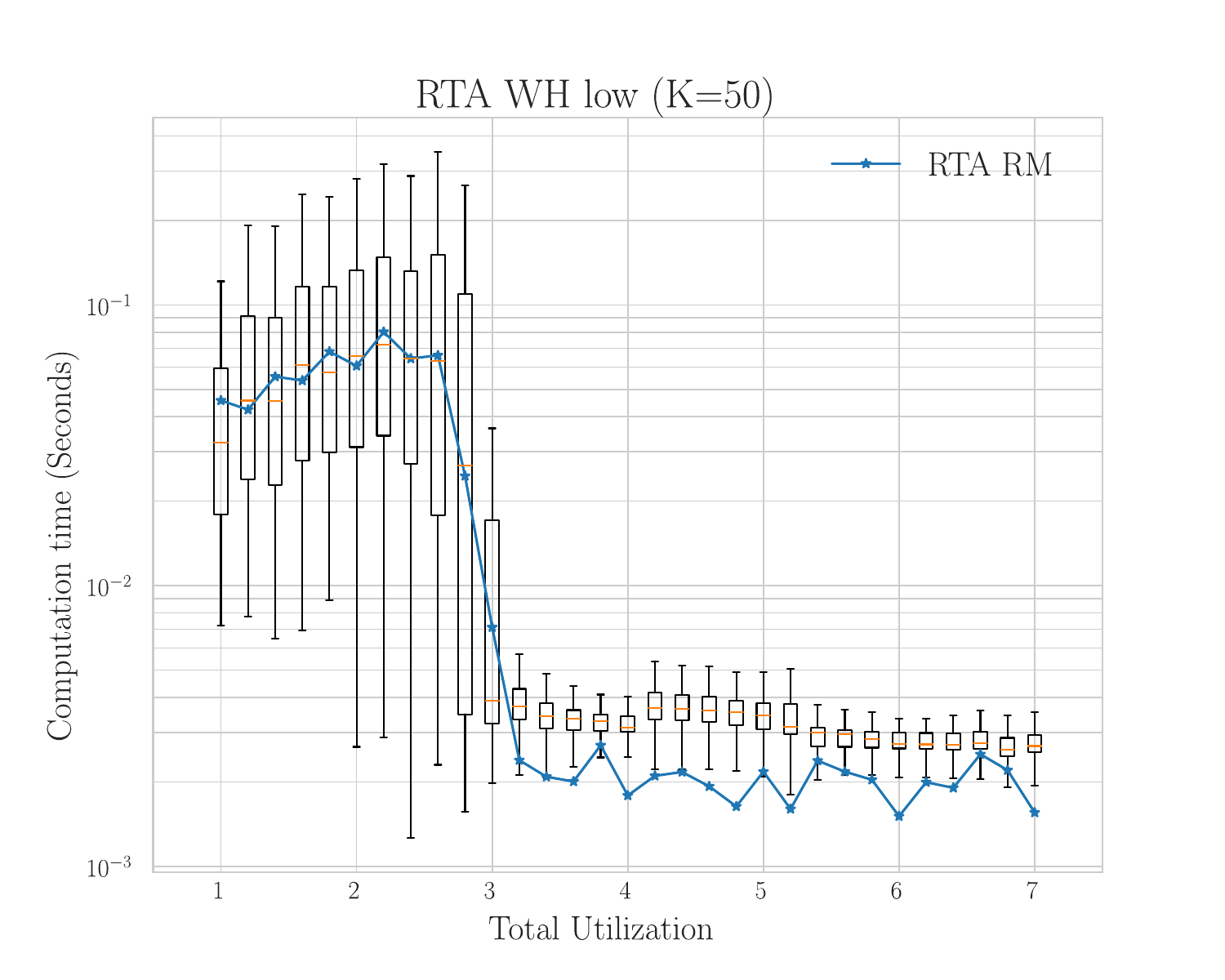} &
        \includegraphics[width=\columnwidth]{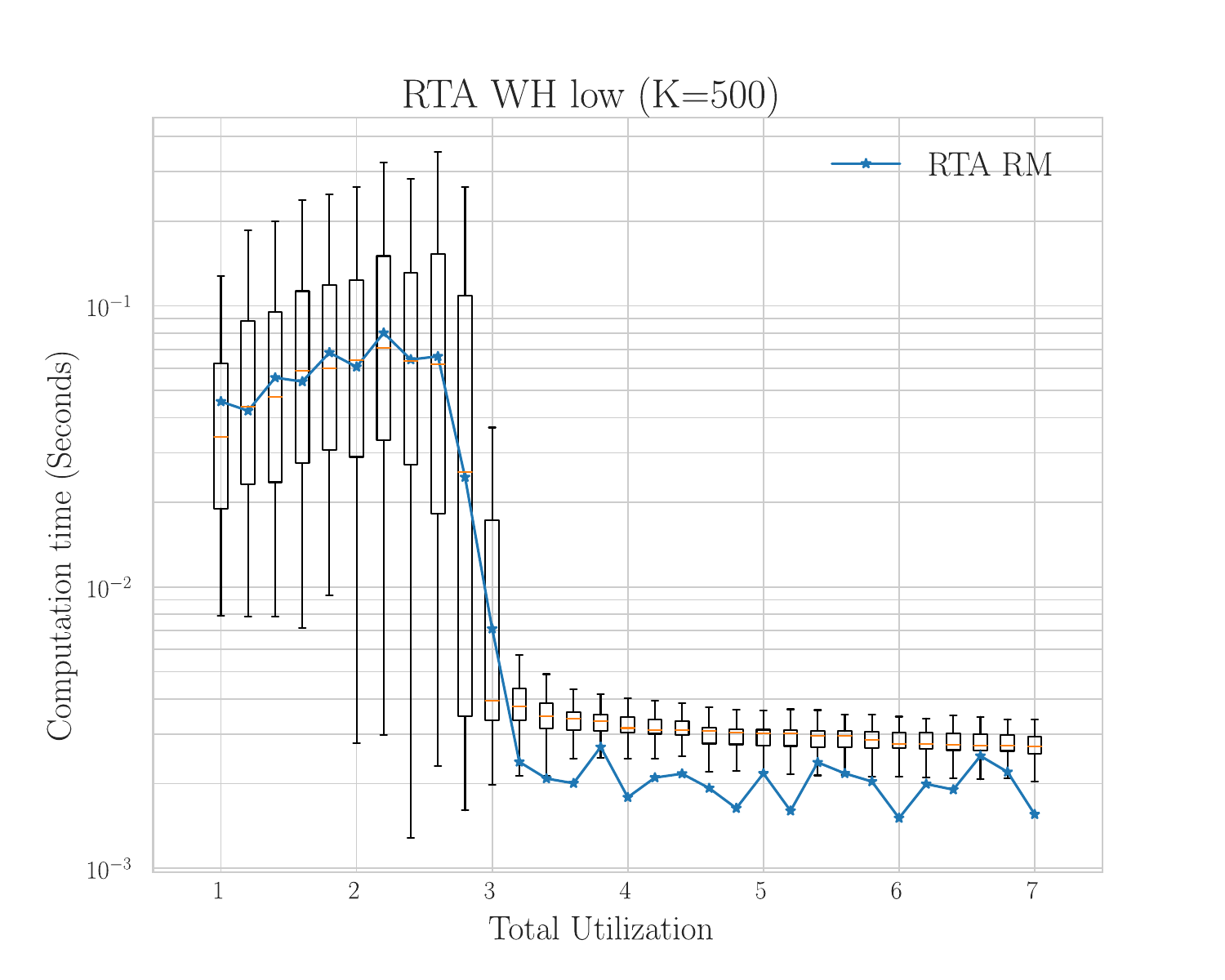} \\
        \includegraphics[width=\columnwidth]{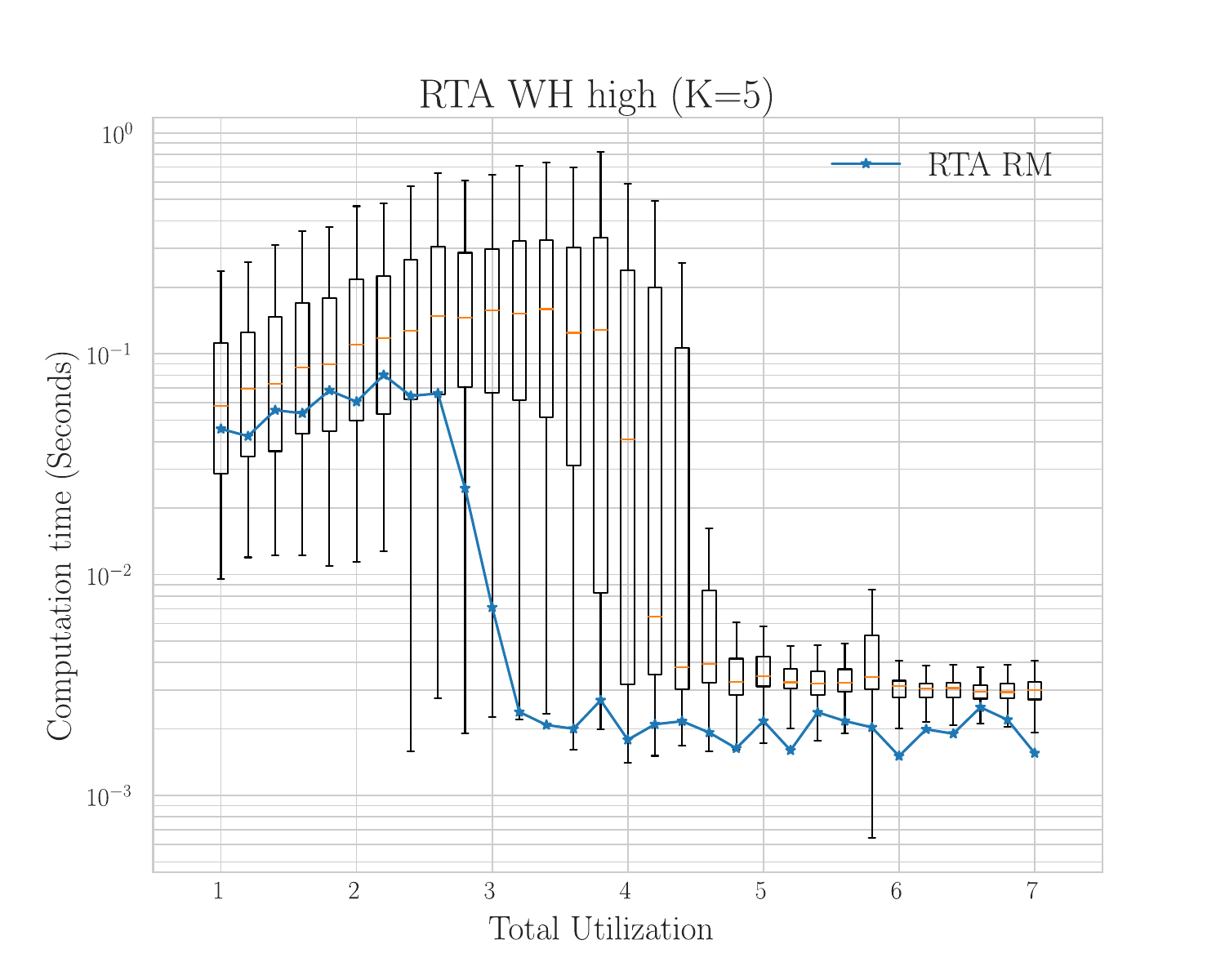} &
        \includegraphics[width=\columnwidth]{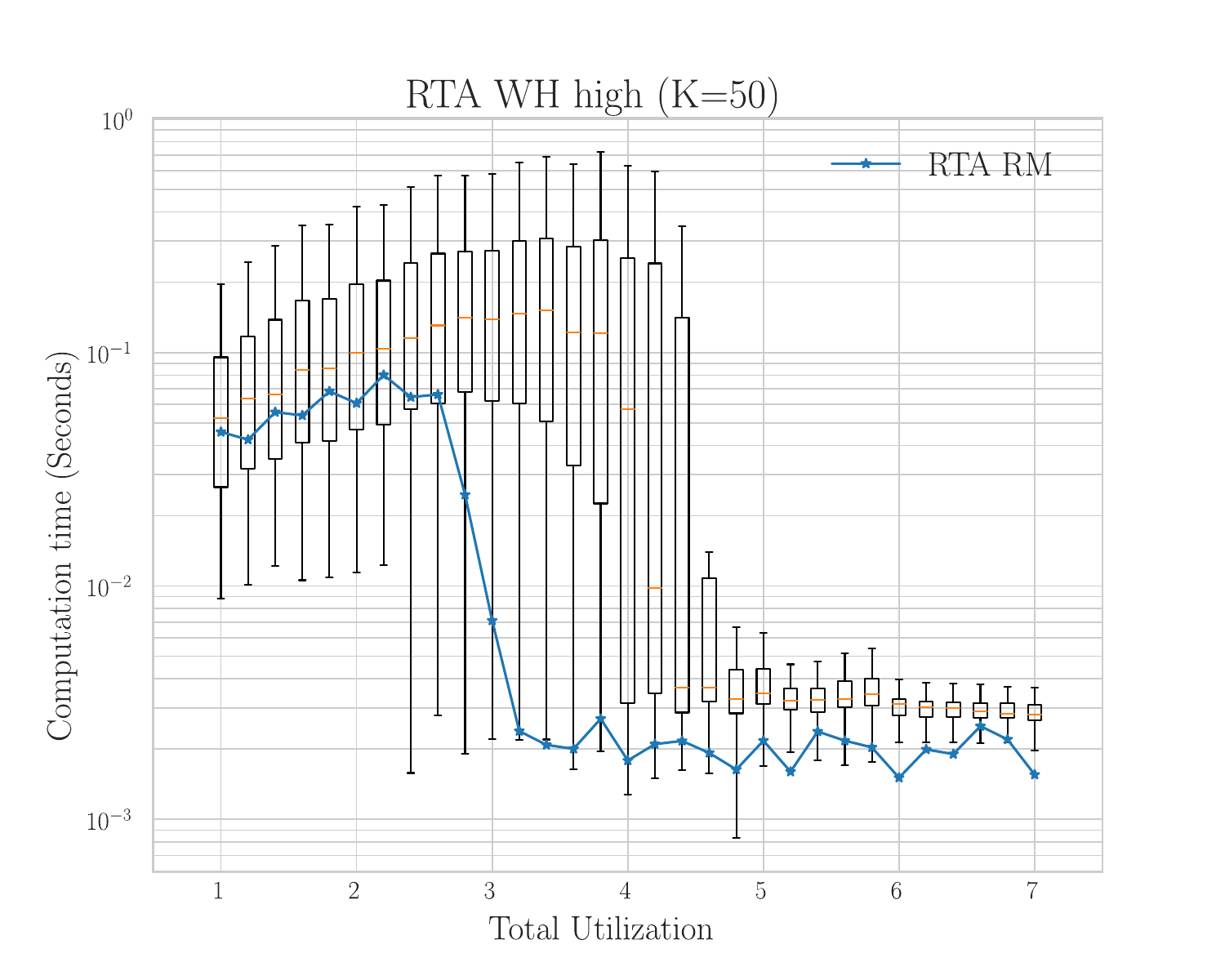} &
        \includegraphics[width=\columnwidth]{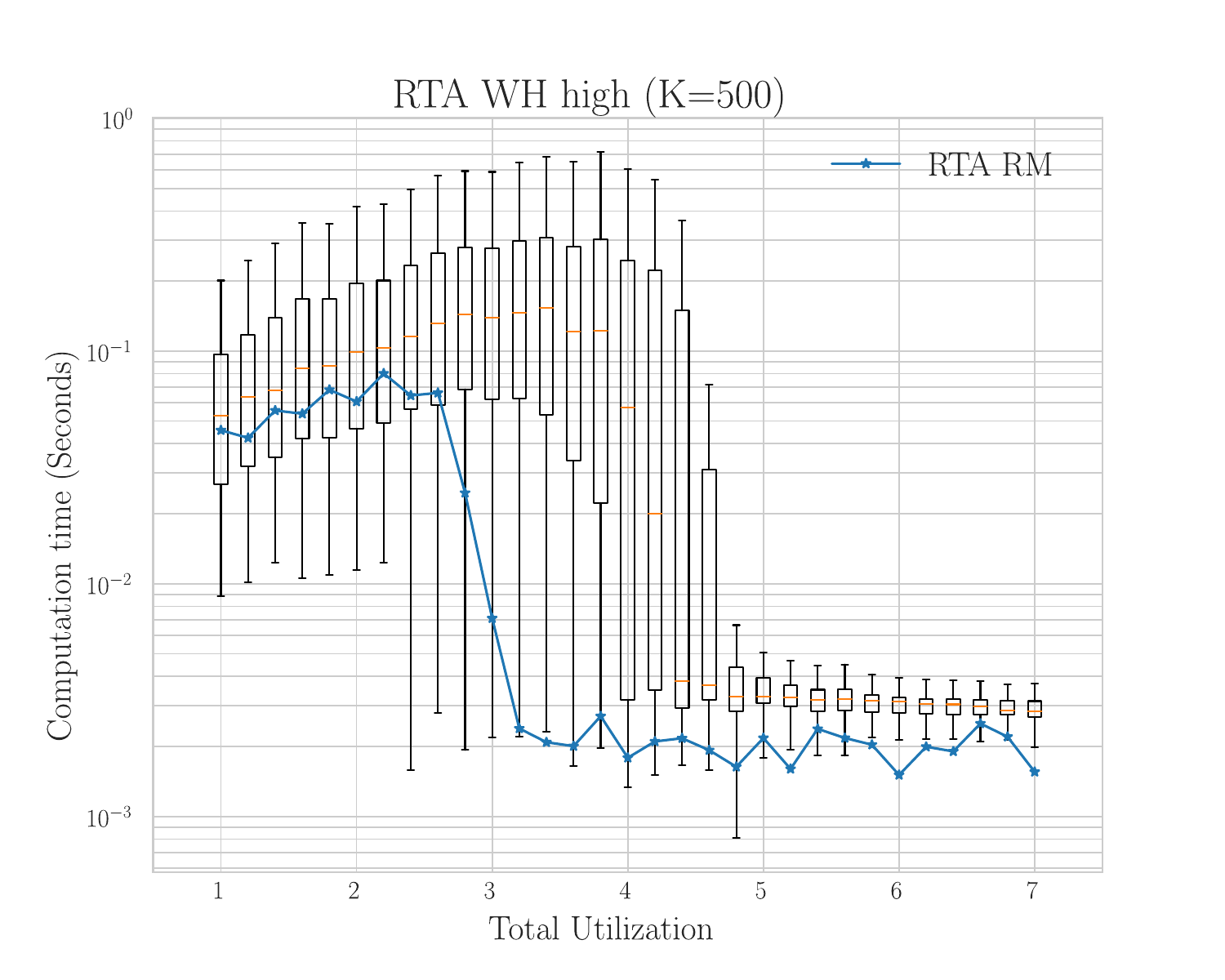} \\
    \end{tabular}
    }
	\caption{Computation time for $K_i$ equals to 5, 50 and 500 (4 cores, 20 tasks in the set).}
    \label{fig:computation-time}
\end{figure*}

\begin{figure*}
    \centering
    \resizebox{2\columnwidth}{!}{
    \begin{tabular}{ c  c  c }
        \includegraphics[width=\columnwidth]{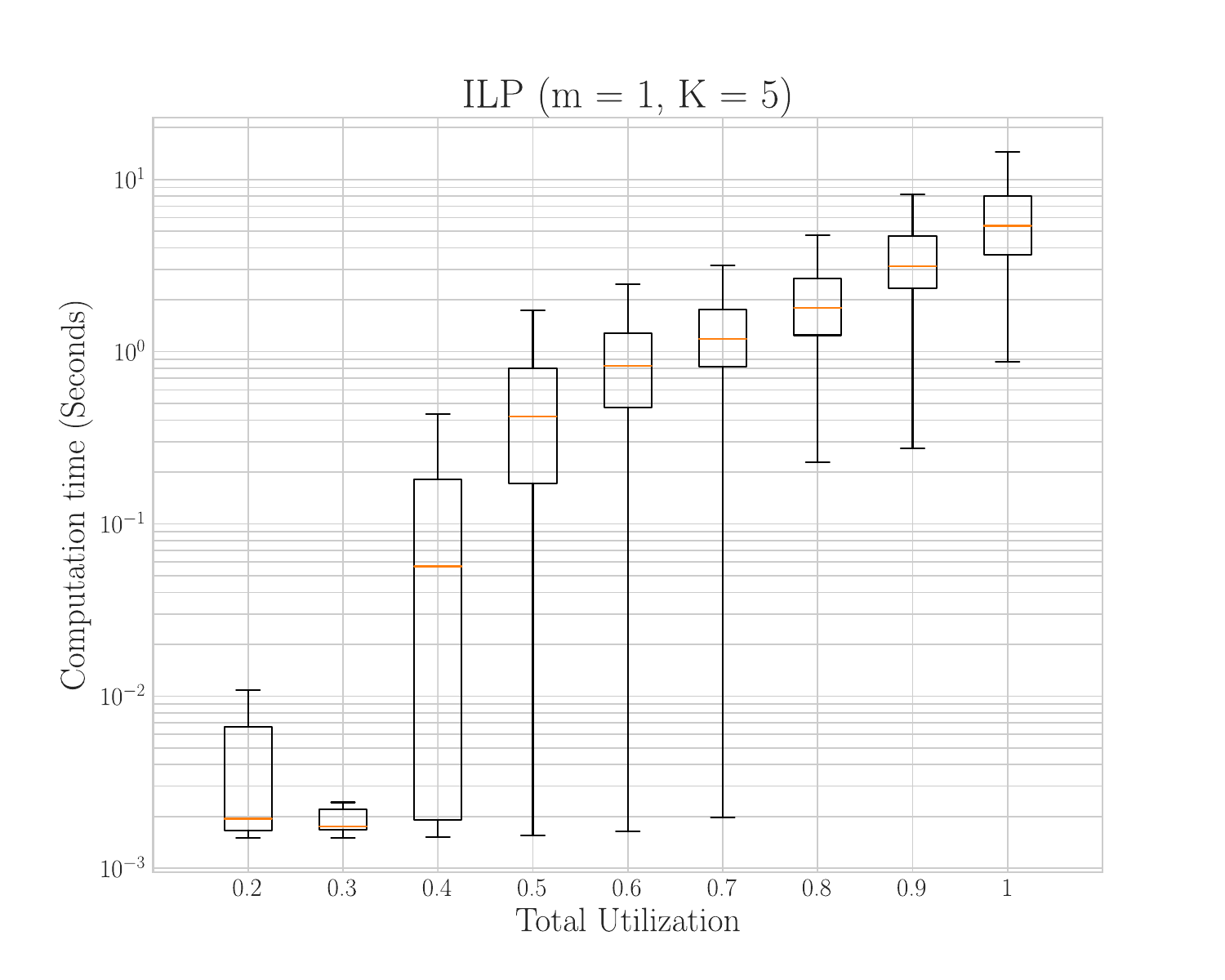} &
        \includegraphics[width=\columnwidth]{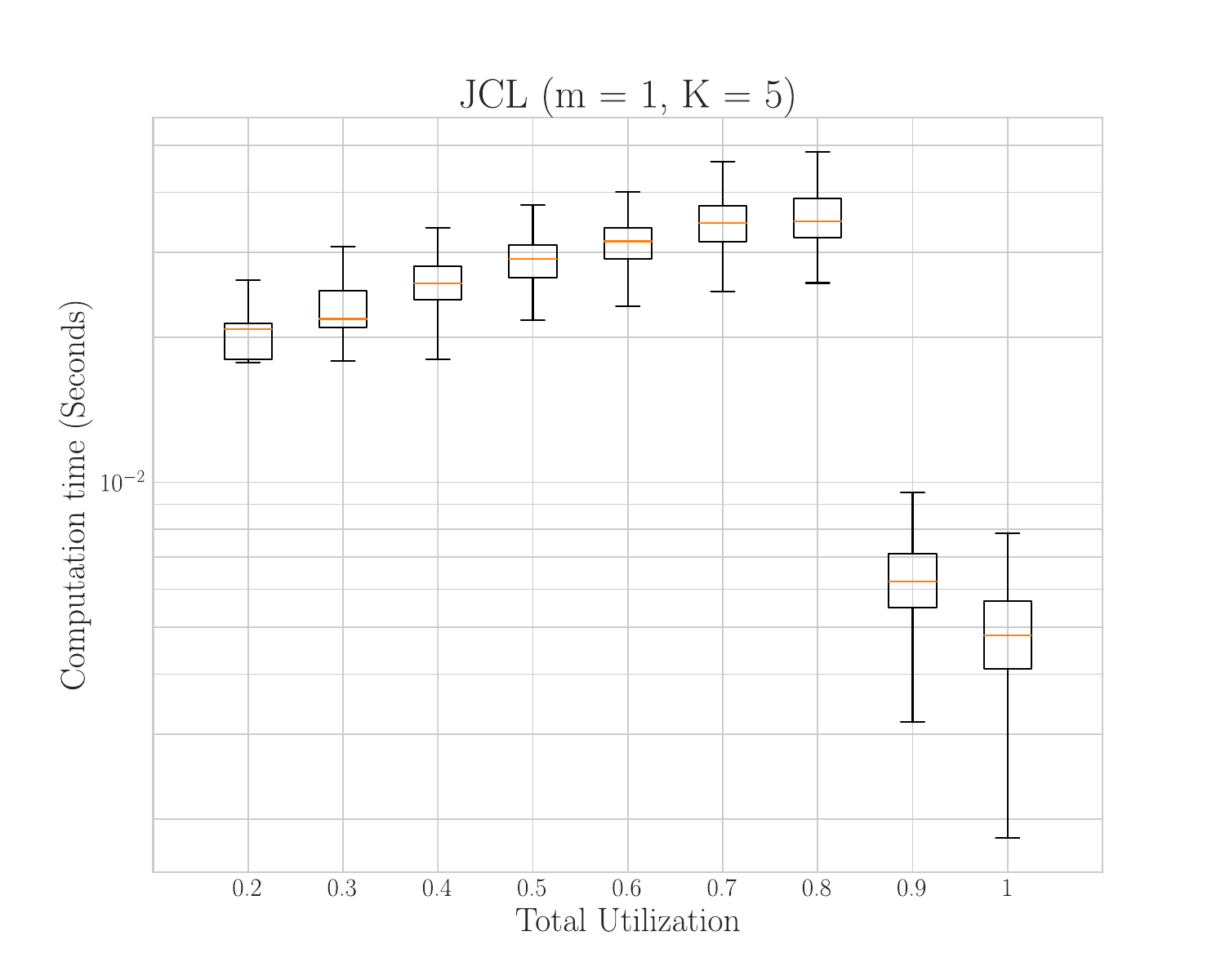} &
        \includegraphics[width=\columnwidth]{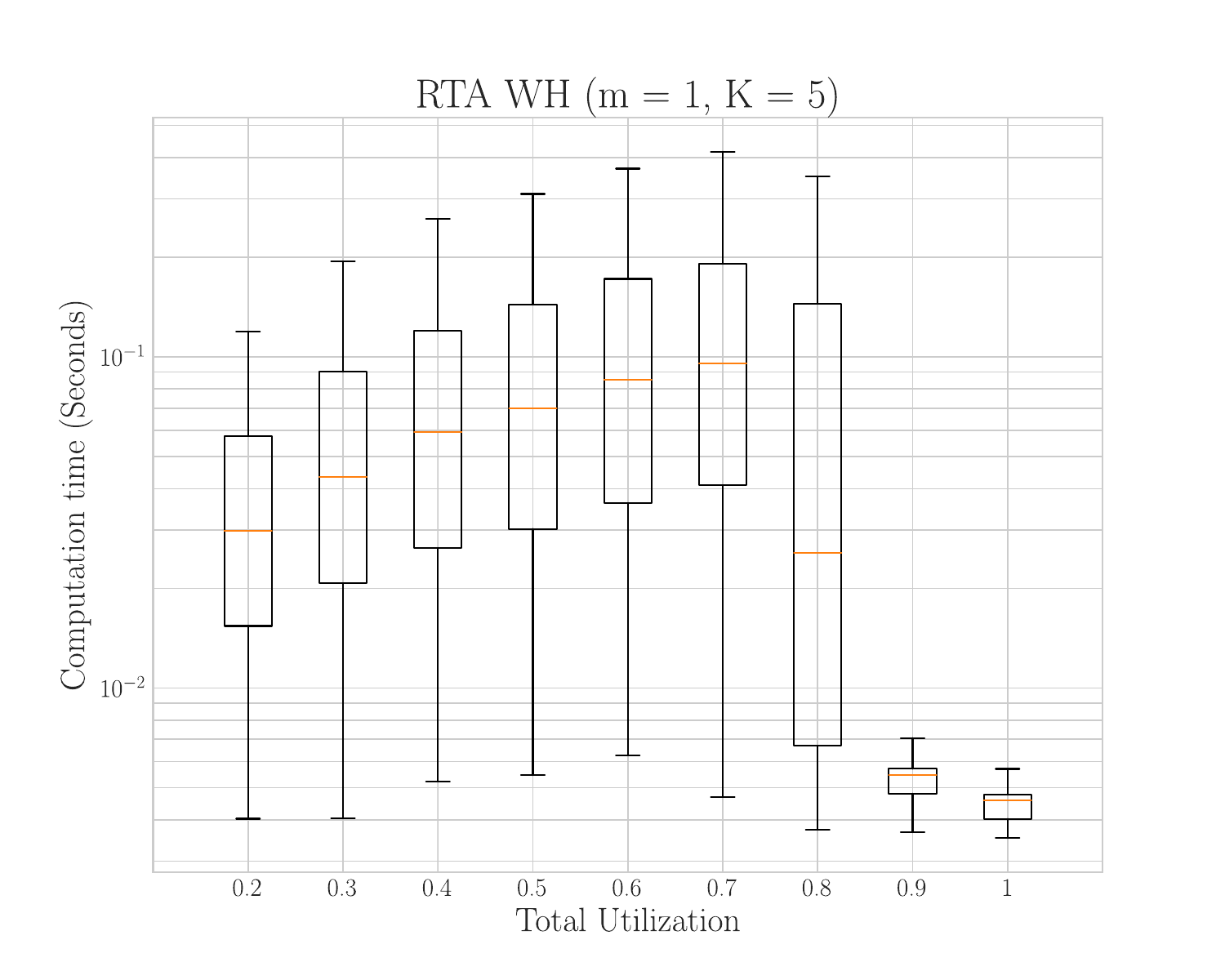} \\
        \includegraphics[width=\columnwidth]{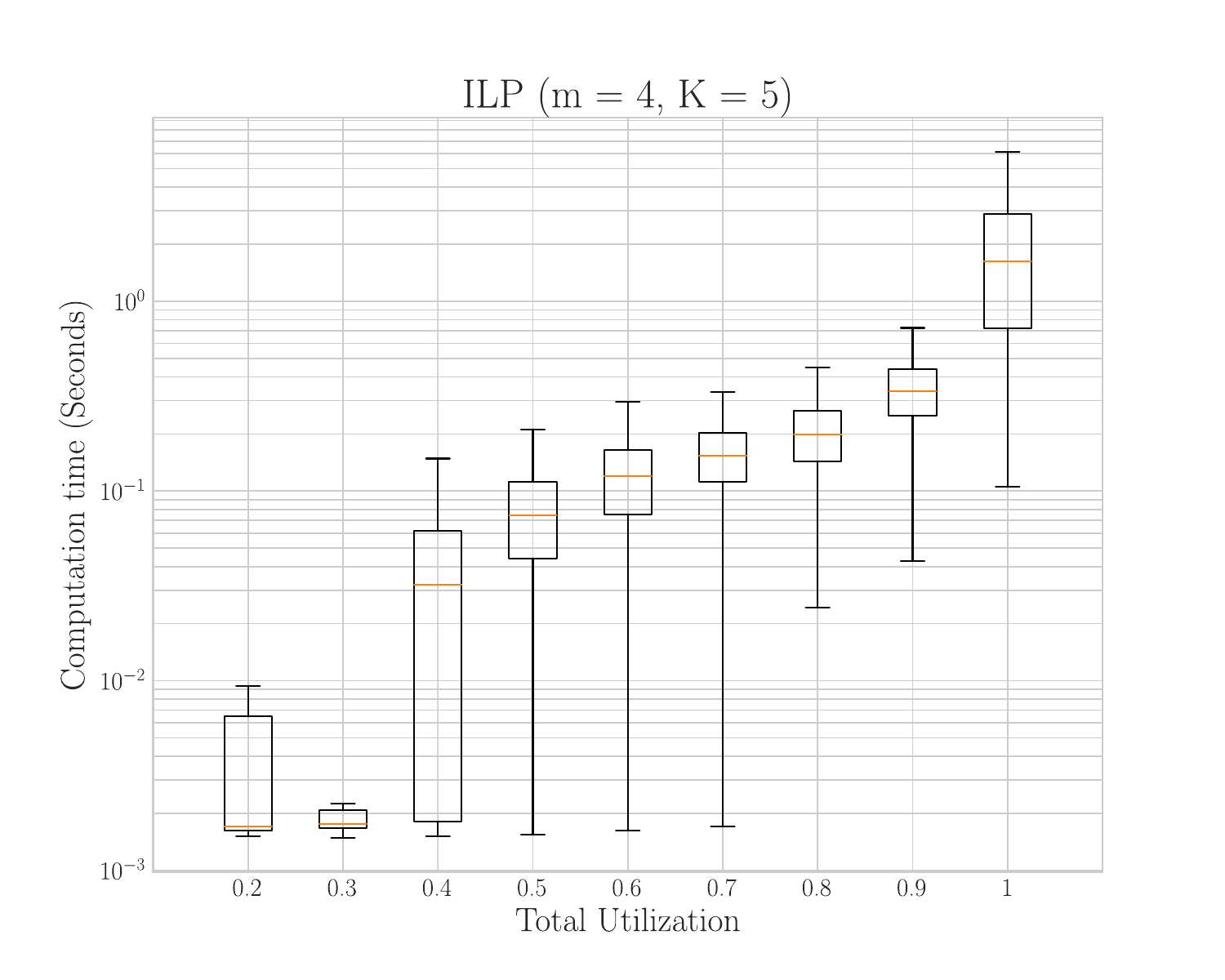} &
        \includegraphics[width=\columnwidth]{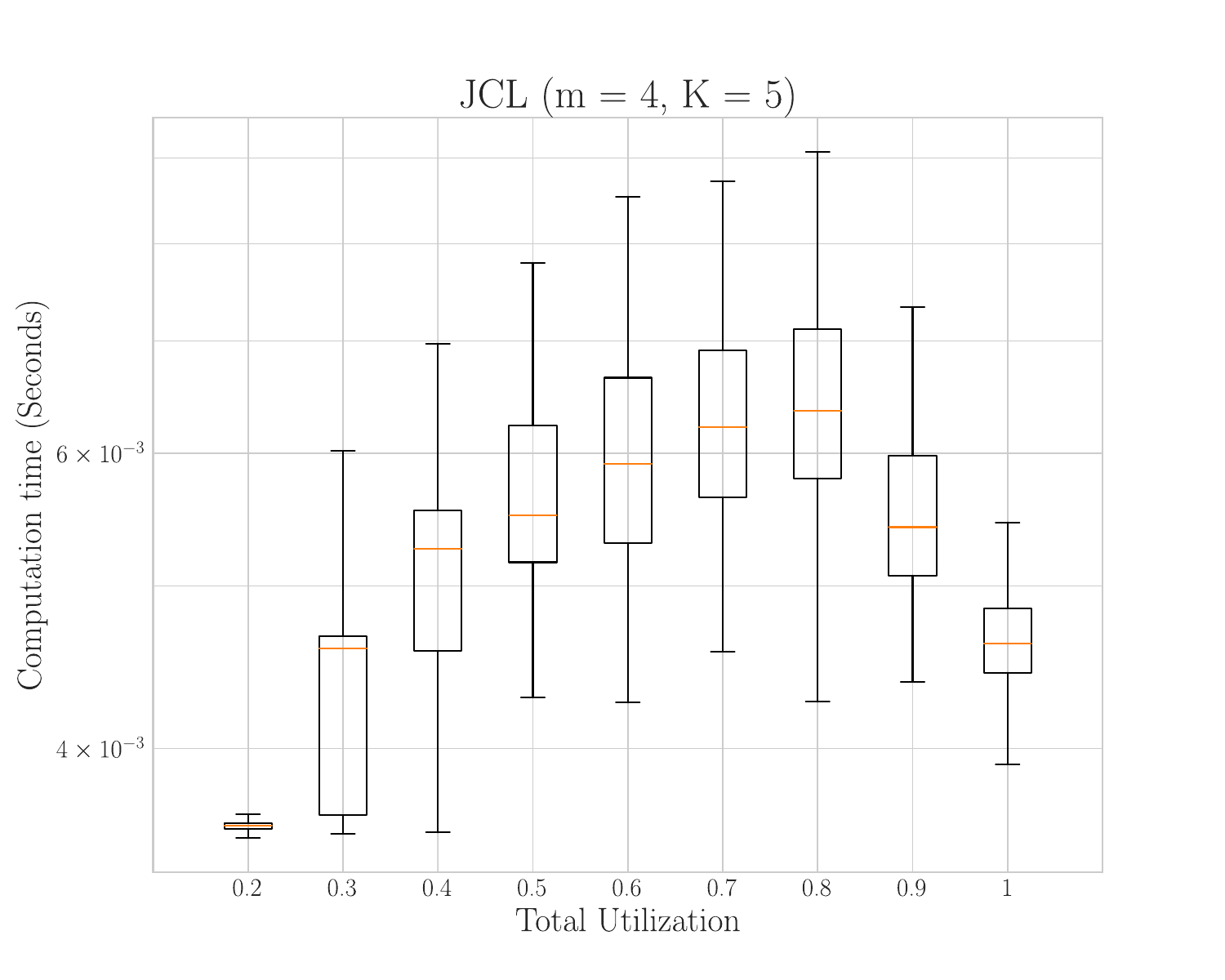} &
        \includegraphics[width=\columnwidth]{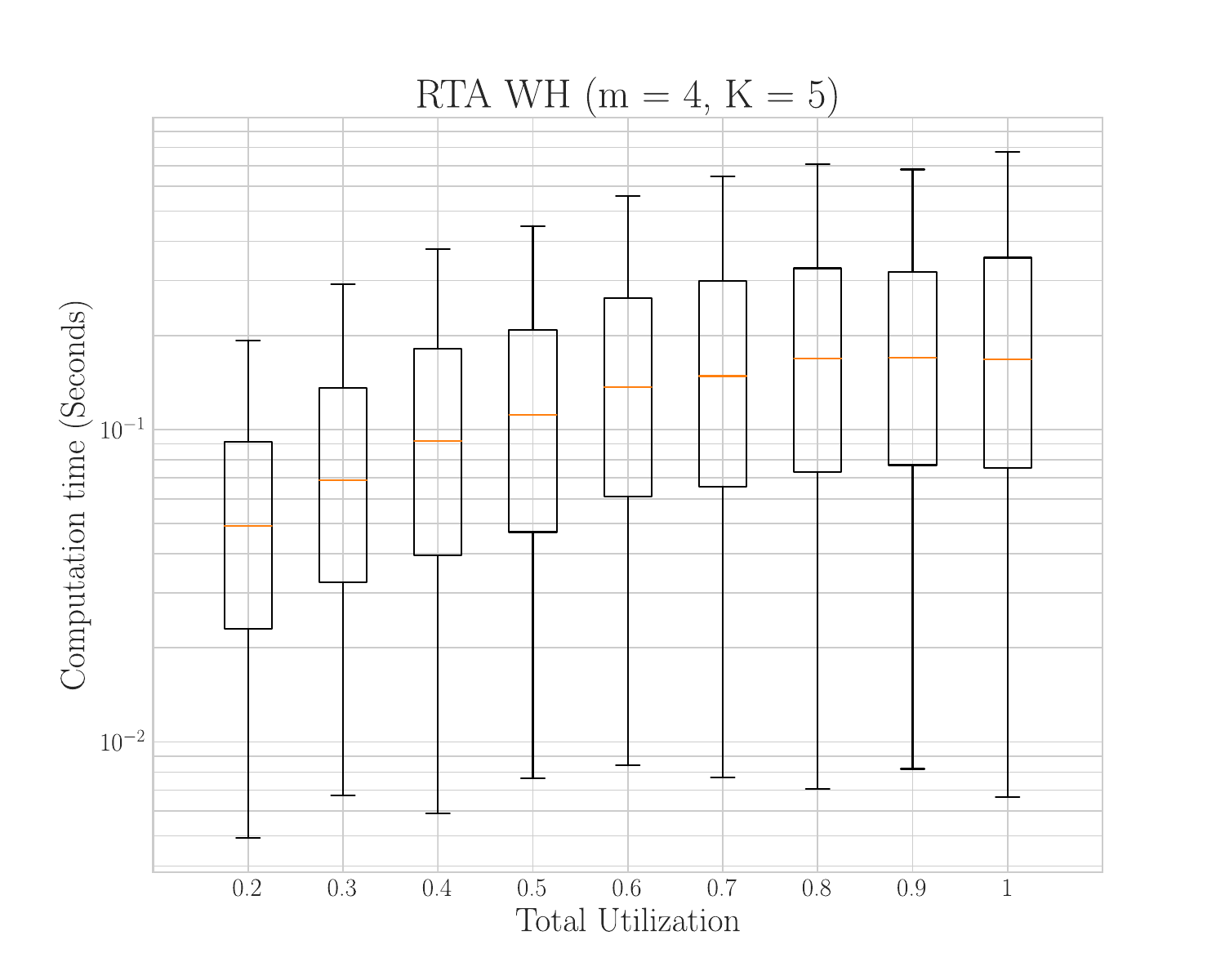} \\
    \end{tabular}
    }
    \caption{Computation time for ILP, JCL and RTA WH (K = 5).}
    \label{fig:computation-time-comparison-ilp}
\end{figure*}

\begin{figure}
    \centering
    \resizebox{1\columnwidth}{!}{
    \begin{tabular}{ c  c }
        \includegraphics[width=\columnwidth]{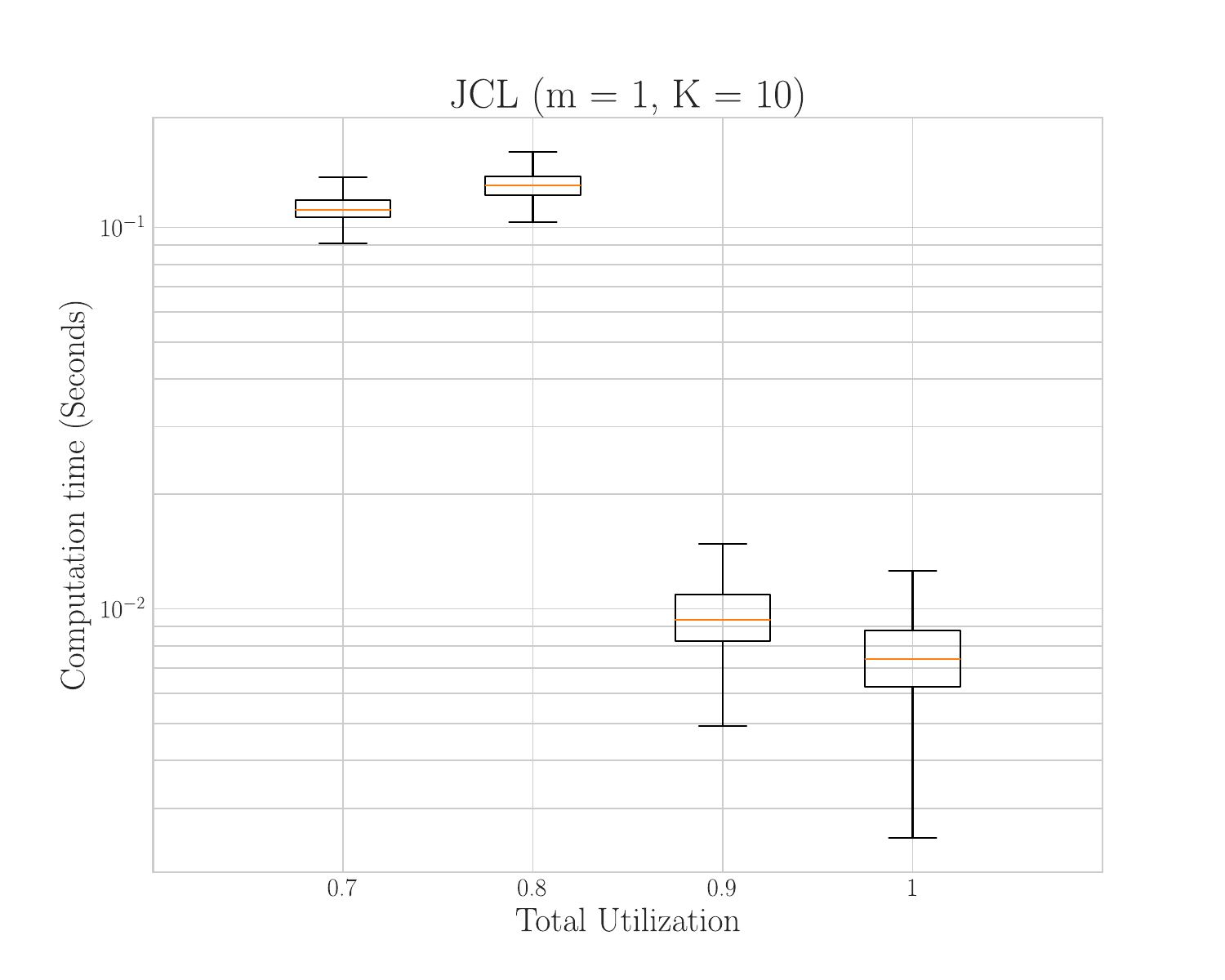} &
        \includegraphics[width=\columnwidth]{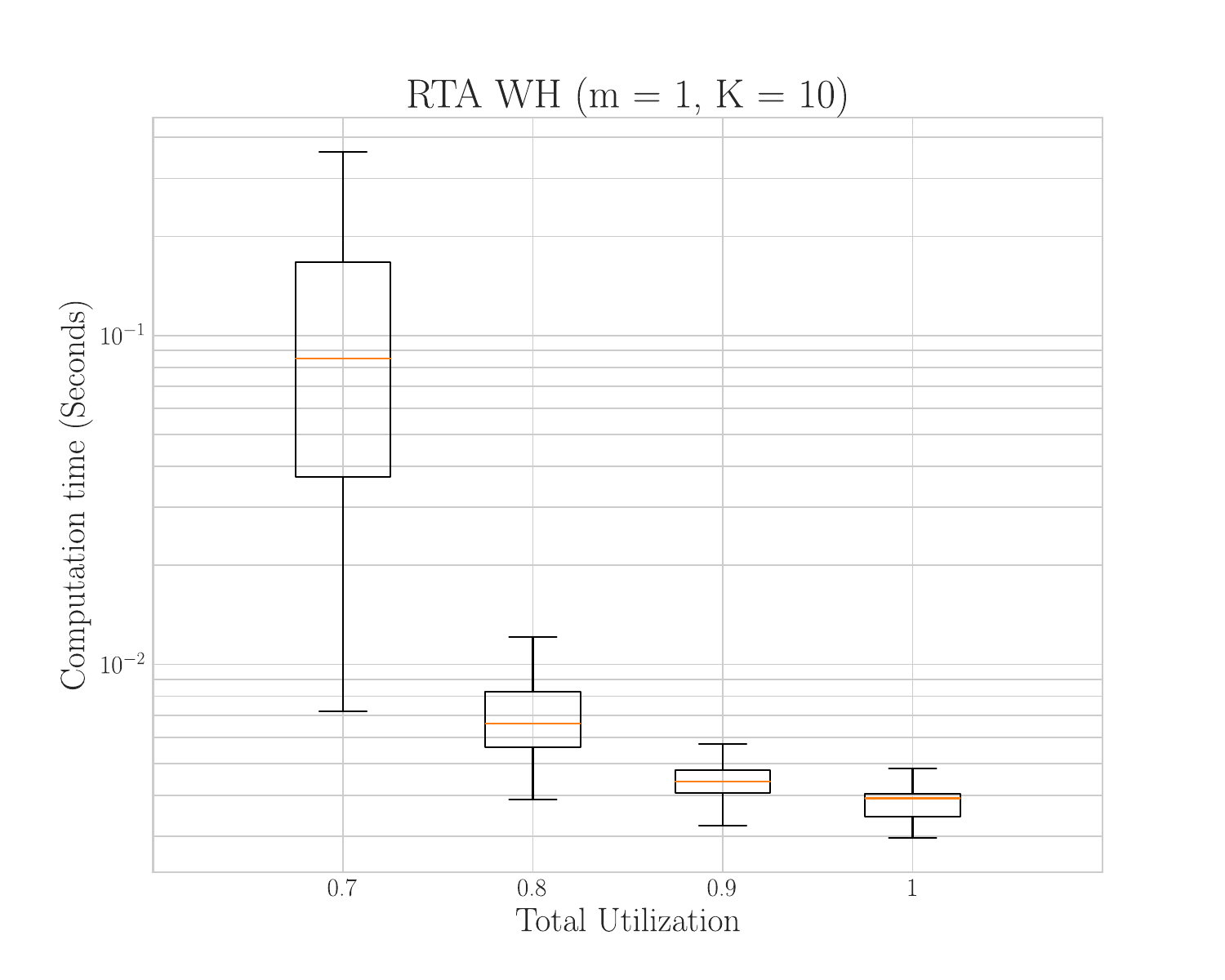} \\
        \includegraphics[width=\columnwidth]{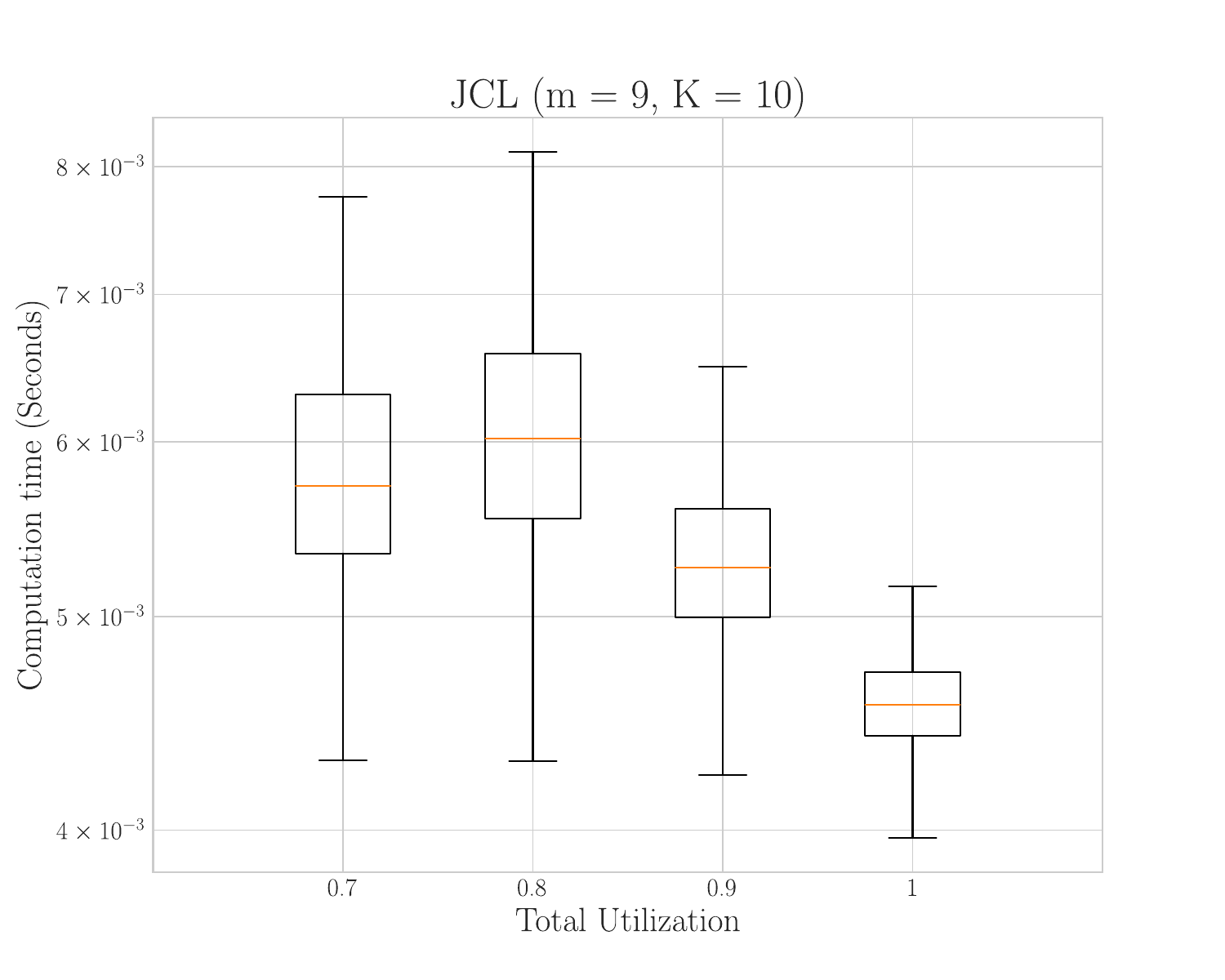} &
        \includegraphics[width=\columnwidth]{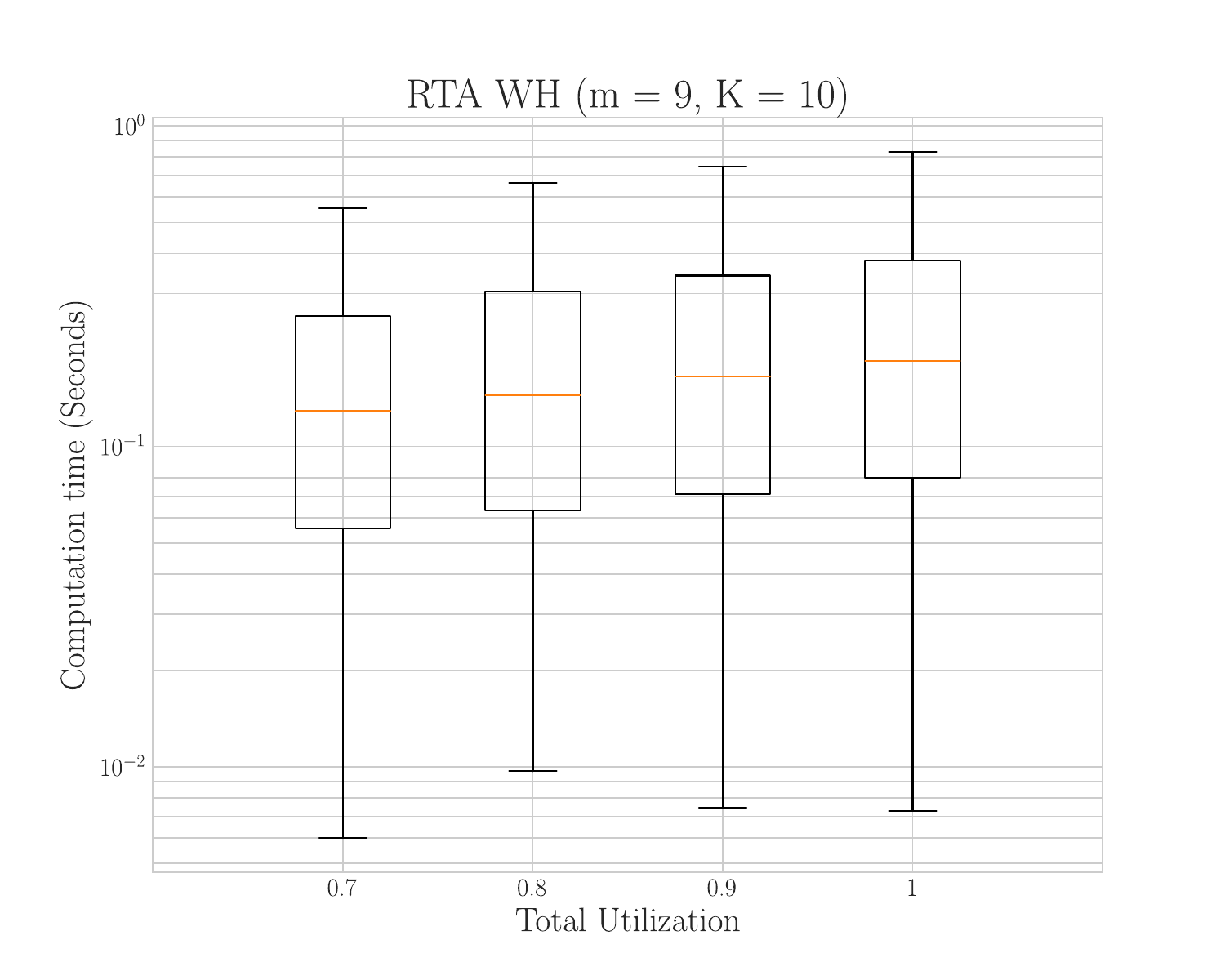} \\
    \end{tabular}
    }
    \caption{Computation time for JCL and RTA WH (K = 10).}
    \label{fig:computation-time-comparison-jcl-10}
\end{figure}

\begin{figure}
    \centering
    \resizebox{1\columnwidth}{!}{
    \begin{tabular}{ c  c }
        \includegraphics[width=\columnwidth]{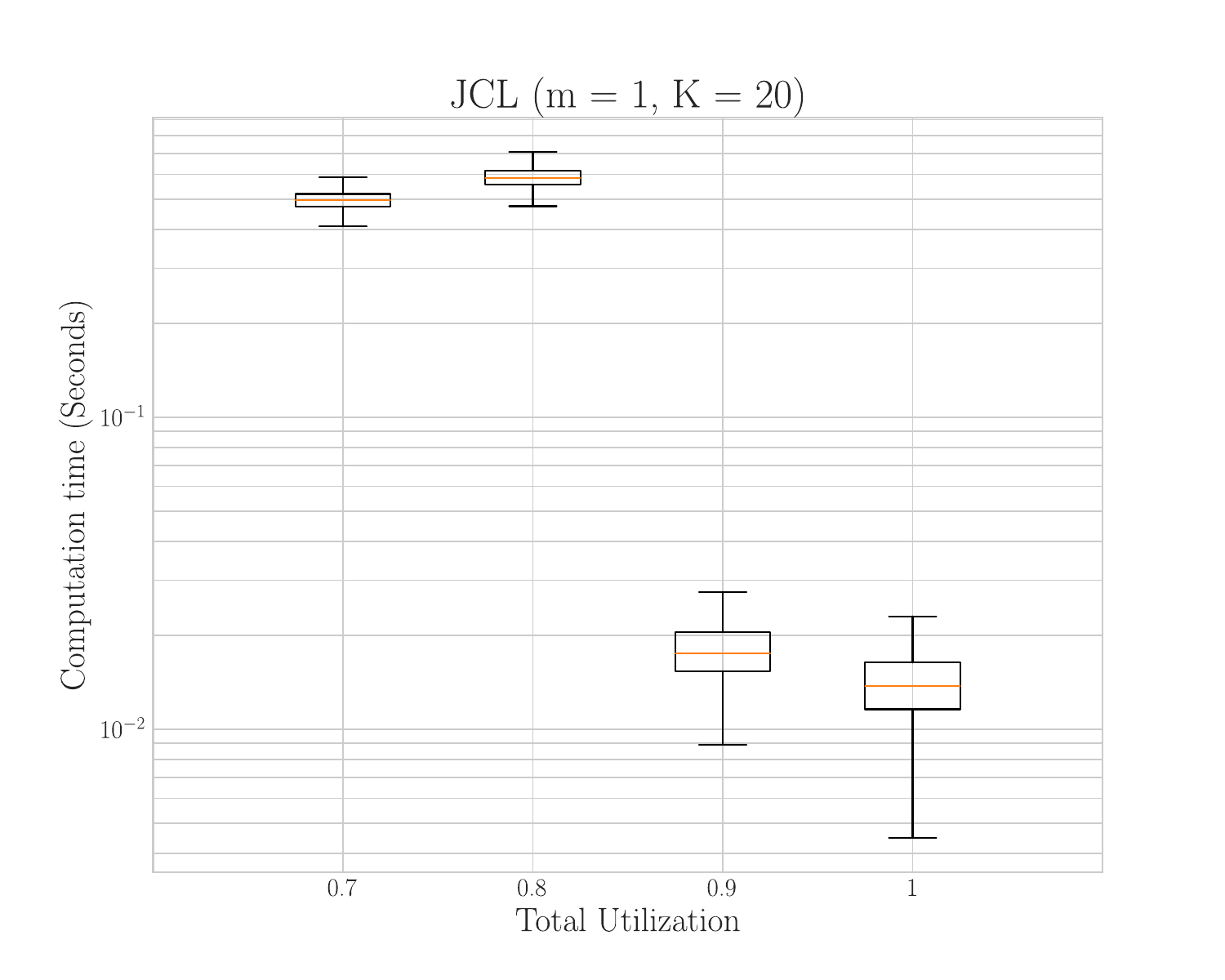} &
        \includegraphics[width=\columnwidth]{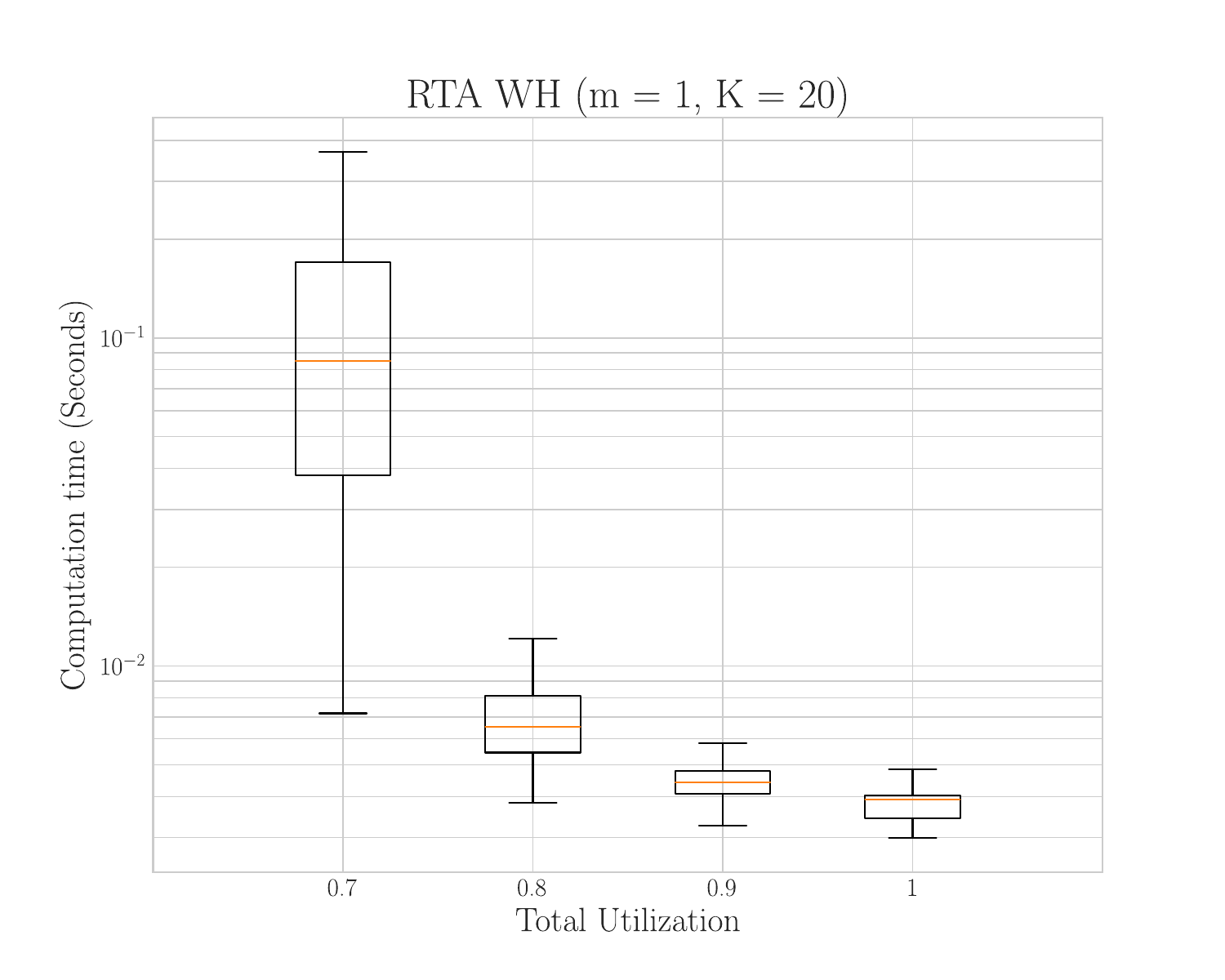} \\
        \includegraphics[width=\columnwidth]{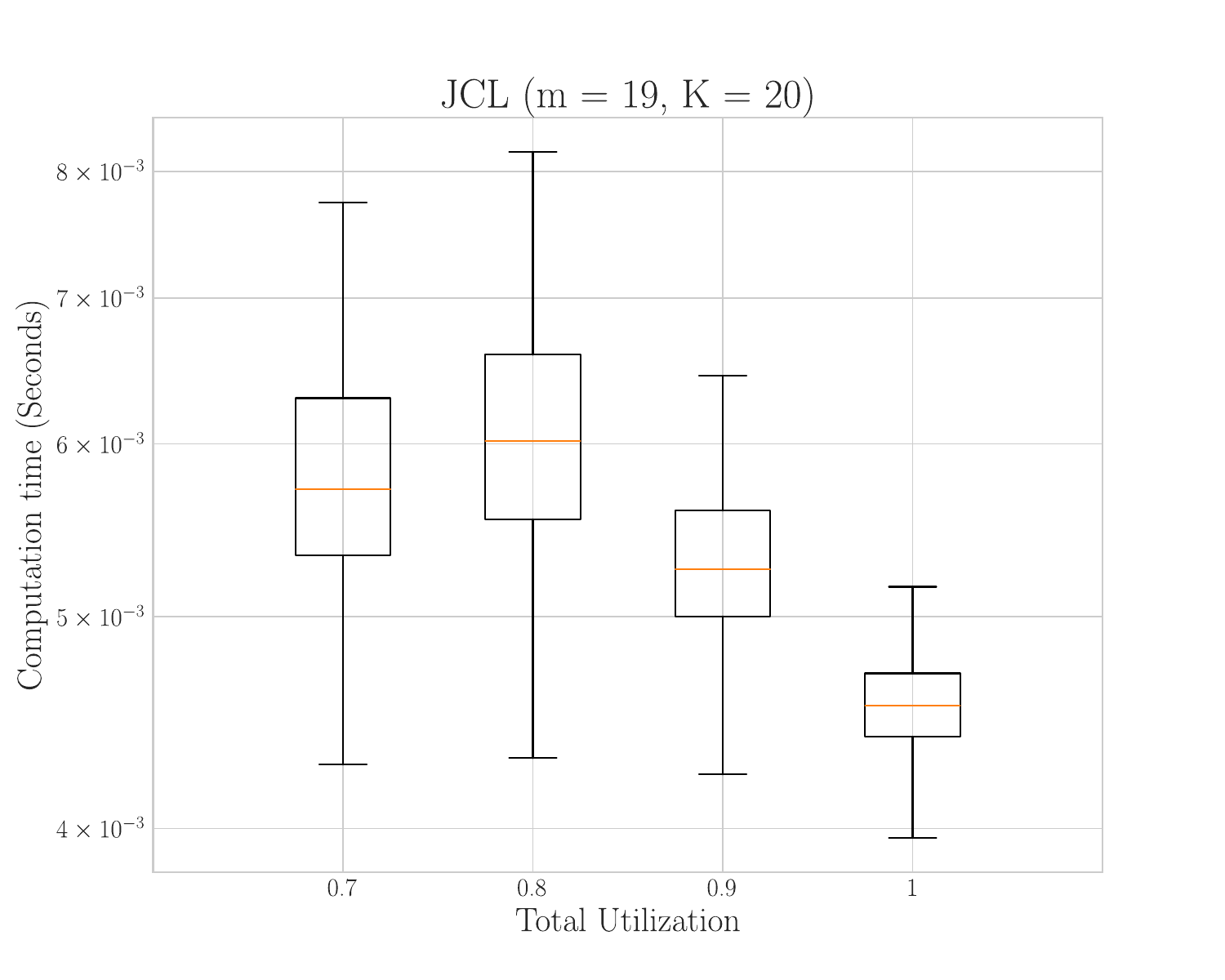} &
        \includegraphics[width=\columnwidth]{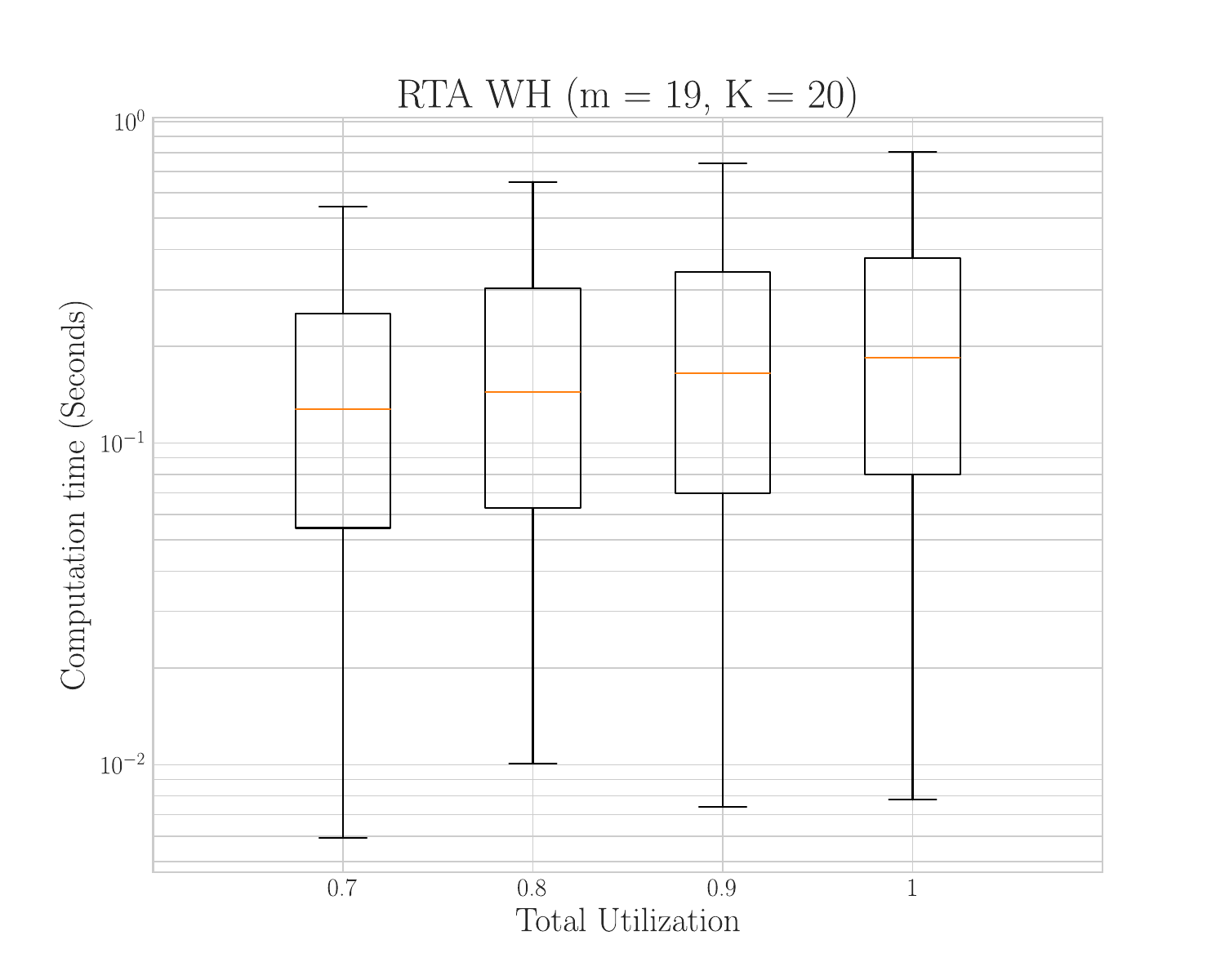} \\
    \end{tabular}
    }
    \caption{Computation time for JCL and RTA WH (K = 20).}
    \label{fig:computation-time-comparison-jcl-20}
\end{figure}

\begin{figure*}[h!]
	\centering
	\resizebox{2.1\columnwidth}{!}{
		\begin{tabular}{ c  c  c }
			(a) $K=5$ &  (b) $K=10$& (c) $K=20$\\ 
			\includegraphics[width=1\columnwidth]{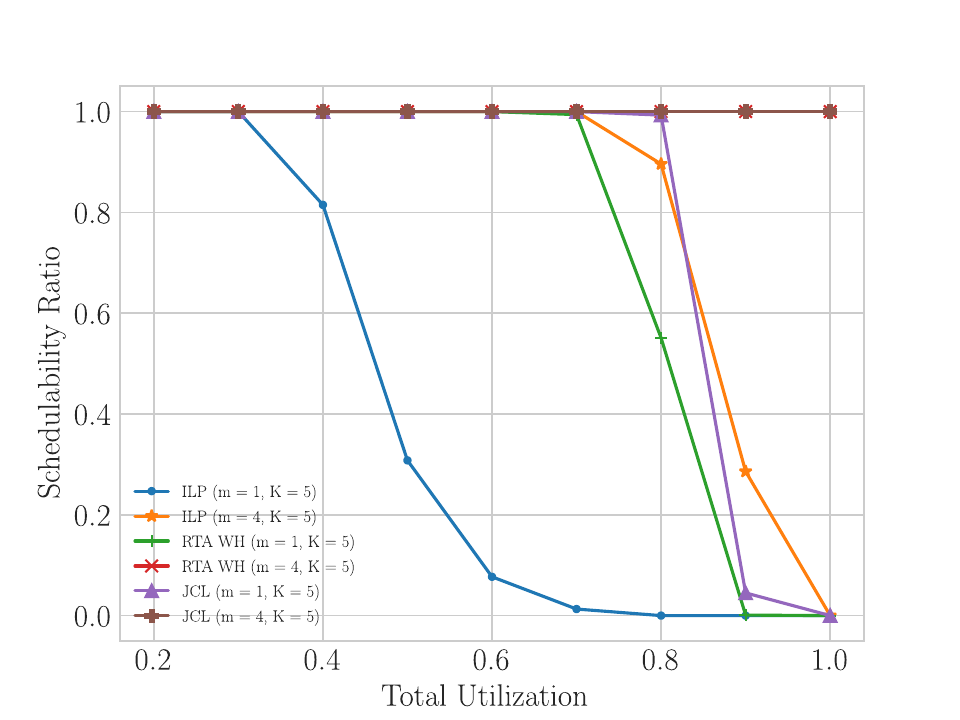} &
			\includegraphics[width=1\columnwidth]{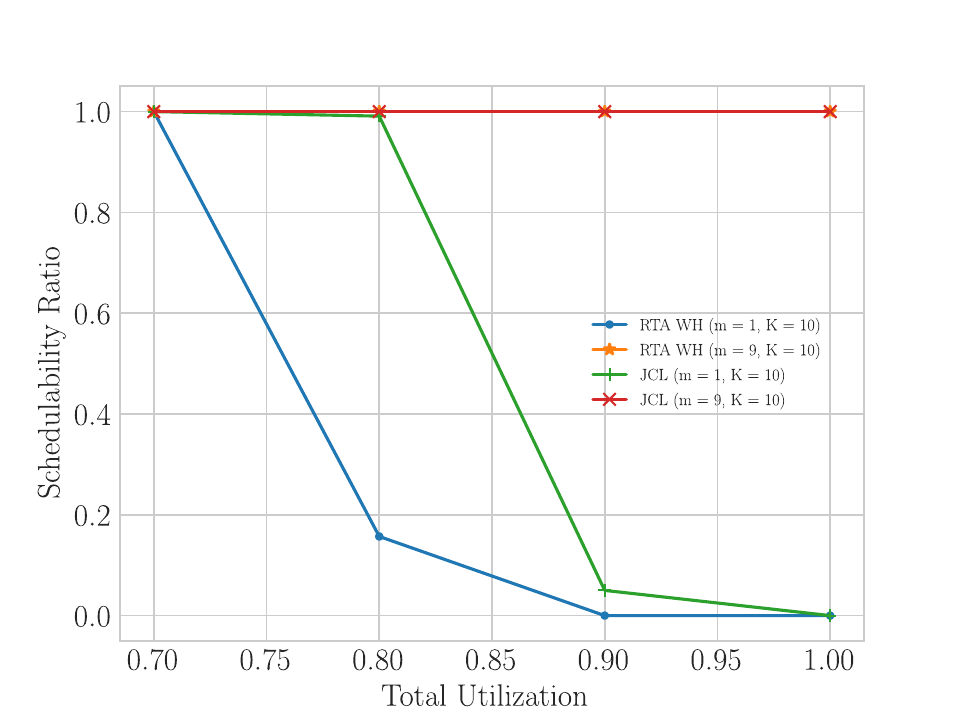} &
			\includegraphics[width=1\columnwidth]{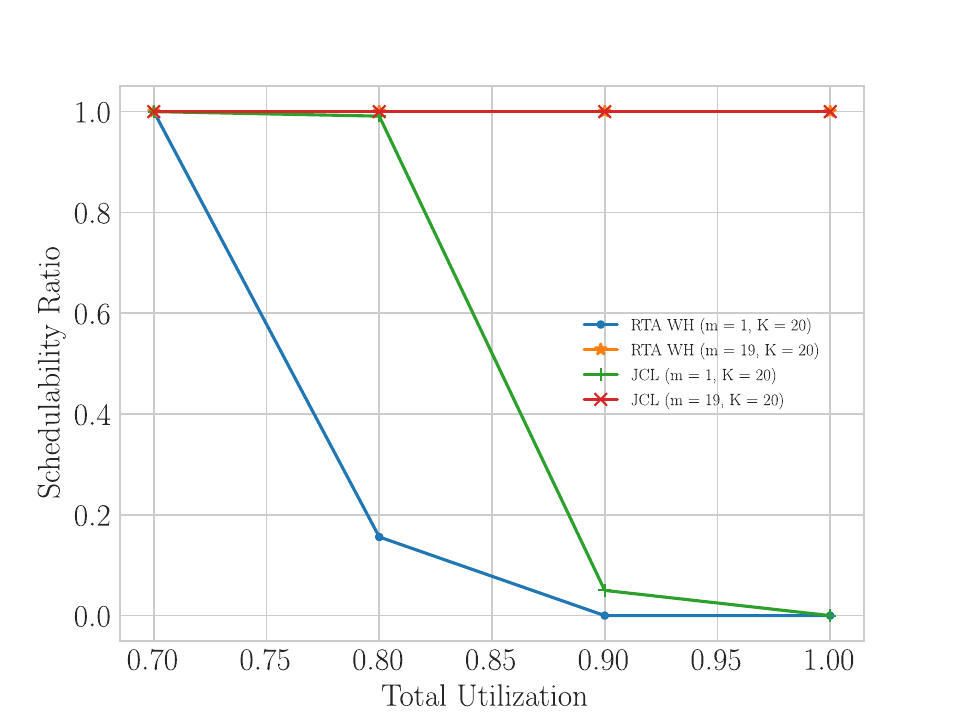} \\
		\end{tabular}
	}
	\caption{Schedulability ratio.}
	\label{fig:sched-ratio-baseline}
\end{figure*}

\subsection{Scheduling Analysis Setup}
Task sets are generated for a list of desired total utilization values using UUnifast.
For every total utilization value, $1000$ sets are generated.
The utilization of the generated tasks is not higher than one.
Furthermore, the generated tasks have implicit deadlines, i.e. their deadlines are equal to their inter-arrival time ($D_i = T_i$).

For every generated task set, the schedulability ratio is calculated and the computation time is measured.
In case of the global scheduling comparison, RTA verifies its schedulability for RM, EDF and the extension of RTA for weakly-hard real-time tasks considering two scenarios.
In the first scenario, all tasks are low-tolerance tasks and in the second, all tasks are high-tolerance tasks.
The values for $m_i$ are chosen randomly between the values that fulfill the desired $m_i/K_i$.
For example, given $K_i = 5$, $m_i$ can be $1$ or $2$ for low-tolerance tasks; and $3$ or $4$ for high-tolerance tasks.

For the comparison against the other weakly-hard scheduling approaches, the value of $m_i$ is $1$ or $K_i - 1$.
Additionally, in case of the experiments which include ILP analysis~\cite{sun2017MILP}, $K_i = 5$ was used due to the scalability issues of ILP analysis~\cite{Natale-ESWEEK17}.
Nevertheless, further experiments compare against JCL using $K_i = 10$ and $K_i = 20$.

Finally, we used a desktop computer with a Intel(R) Core(TM) i7-8700 processor (6 cores, 2 threads per core) and 32GB of RAM running Ubuntu to run the experiments.

\subsection{Scheduling Analysis Experiments}

\figurename~\ref{fig:sched-ratio}.(a), (b), (c) show the schedulability ratio using a set of $20$ tasks for $2$, $4$ and $8$ cores, respectively.
The plots for EDF show a worst schedulability ratio in comparison with RM.
This is due to RTA calculates a higher interference for EDF than for RM (remember that in RM only tasks with higher priority are considered).
Furthermore, it is observed that the schedulability ratio for the weakly-hard tasks is much better than for RM.
Additionally, high-tolerance tasks seem to be schedulable even after the practical limit ($U =$~number of cores).
This is explained by Lemma~\ref{lem:workload-equivalent}, which allows us to consider the workload coming from the high-tolerance tasks as it were coming from a task with longer inter-arrival time, reducing the utilization of the task.
This also explains why this behavior is not seen for low-tolerance tasks.

\figurename~\ref{fig:sched-ratio-tasks} shows the schedulability ratio for $20$, $50$ and $100$ tasks in the set when scheduled on $4$ cores.
Here, it is observed that sets with more tasks have a better schedulability.
The reason is that having more tasks, while keeping the same total utilization, makes the tasks more lightweight which reduces the interference.
On the other hand, the computation time is increased with the number of tasks, as it is seen in \figurename~\ref{fig:computation-time-tasks} (the curve shows the average duration for RM).
Also, note that computation time decreases with the schedulability ratio because the scheduling analysis ends when the first no schedulable task is found.

\figurename~\ref{fig:sched-ratio-k} shows the schedulability ratio of tasks scheduled on $4$ cores when $K_i$ is $5$, $50$ and $500$.
There are no significant changes in the schedulability ratio when the value of $K_i$ is changed.
The same is observed in terms of computation time, see \figurename~\ref{fig:computation-time}.

\subsection{Comparison Against Single-Core Scheduling Approaches}

\figurename~\ref{fig:sched-ratio-baseline}.(a) shows the schedulability ratio comparison between ILP, JCL and RTA WH using task sets with $30$ tasks.
For experiments with $m_i = 1$, JCL is the best of the tree while ILP is the worst of them.
The bad results of the ILP approach are explained by remembering that the ILP approach works at task-level while the other approaches work at job-level.
In case of RTA WH, there is pessimism introduced by using a multi-core scheduling analysis for single-core.

When $K_i = 10$ and $K_i =  20$, \figurename~\ref{fig:sched-ratio-baseline}.(b) and (c) show that the schedulability ratio is also better for JCL than for RTA WH.
Again, this is explained similar as before, i.e.\ using RTA for single-core introduces pessimism.
However, the schedulability ratio curves for $m_i = K_i - 1$ are similar for JCL and RTA WH until a total utilization equals to $1$.

\figurename~\ref{fig:computation-time-comparison-ilp} shows the computation times for ILP, JCL and RTA WHA.
For ILP, the computation time increases with the utilization.
However, for JCL and RTA WH, the computation time decreases together with the schedulability ratio.
Furthermore, the JCL approach is faster than the others when $K_i = 5$.
Nevertheless, while increasing $K_i$ to $10$, the computation times of JCL are in the same range than those for RTA WH (see \figurename~\ref{fig:computation-time-comparison-jcl-10}).
When $K_i = 20$, RTA WH is faster than JCL (see \figurename~\ref{fig:computation-time-comparison-jcl-20}).
In fact, RTA WH does not variate too much while changing the value of $K_i$.
This was also observed in \figurename~\ref{fig:computation-time}.

\subsection{Priority Assignment Overhead}

\begin{figure}[t!]
    \centering
    \includegraphics[width=.67\columnwidth]{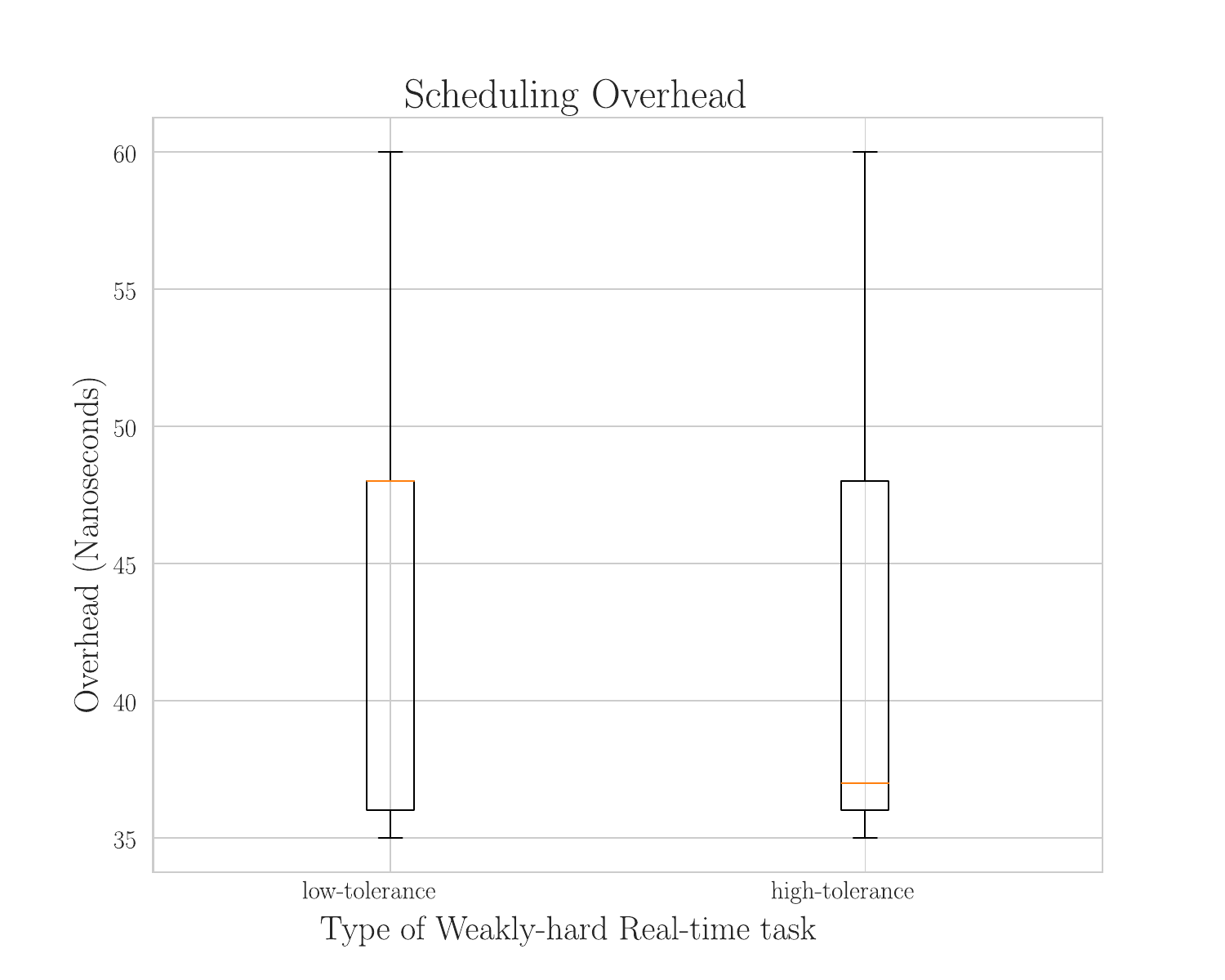}
	\caption{Scheduling overhead distribution (1000 samples).}
    \label{fig:sched-overhead}
\end{figure}

In order to know the scheduling overhead due to the priority assignment algorithm defined in Definition~\ref{def:job-level}, we have conducted an experiment for measuring the execution time.
This experiment was executed on the same processor architecture used in the on-board computers of CALLISTO, i.e. Intel Atom Quad-Core with core frequency of $1.91GHz$.
The real-time operating system used in this experiment is RTEMS $5.1$ and the code was compiled with optimization level 2.

\figurename~\ref{fig:sched-overhead} shows the execution time distribution for executing the priority assignment in case of both low and high-tolerance tasks.
The algorithm execution was repeated $1000$ times for every type of task.
Results show an overhead below $60$ nanoseconds which is negligible overhead comparing to the execution time of the tasks.
Low-tolerance tasks show a higher median than high-tolerance tasks.
This behavior is related with the fact that high-tolerance tasks tolerate more deadline misses before going back to the $\mathcal{JC}^{0}_{i}$, making some part of the algorithm be executed more often.

\section{Conclusion}
\label{sec:conclusion}
In centralized embedded system architectures, multi-core processors are utilized to consolidate multiple tasks leveraging the high computational power and the low power-consumption of multi-core platforms.
In real-time systems, if few deadline misses are tolerable, leveraging the weakly-hard model can reduce the over-provisioning, hence, consolidating more real-time tasks on to the multi-core platform.
This paper proposed a job-class based global scheduling for weakly-hard real-time tasks, appended with a schedulability test to compute the needed guarantees.
The scheduling algorithm exploits the tolerable deadline misses by assigning different priorities to jobs upon urgency of meeting their deadline.
Such job-level priority assignment reduces the interference with low-priority tasks and helps them to satisfy their weakly-hard constraints.
The proposed schedulability analysis utilizes neither ILP nor reachability tree-based analysis, as similar approaches in the literature.
Rather, it focuses on verifying the schedulability of the maximum tolerable consecutive deadline misses.
Our experiments show that the proposed analysis is schedulable where through all the synthetic test cases, the computation time for the proposed analysis does not exceed $1.6$ seconds for $100$ tasks and $4$ cores.
Also, results illustrates that the improvement on the schedulability ratio is up to 40\%  over the global Rate Monotonic (RM) scheduling and up to 60\% over the global EDF scheduling.
Furthermore, our future work will contemplate a reduction in the limitations for low-tolerance tasks, the exploration of the priority assignment in the context of job-class-level scheduling and the interference between tasks due to shared resources.

\bibliographystyle{IEEEtran}
\bibliography{bibliography}

\appendix
\section{Lemma~\ref{lem:harder-constraint} Proof}
\label{sec:theorem-proof}

Theorem~5 of~\cite{bernat2001weakly} states that a weakly-hard constraint $\big( \begin{smallmatrix} a \\ b \end{smallmatrix} \big)$ is harder than (denoted as $\preccurlyeq$) other constraint $\big( \begin{smallmatrix} p \\ q \end{smallmatrix}\big)$ if:
\begin{equation}
    \left( \begin{array}{c} a \\ b \end{array}\right) \preccurlyeq \left( \begin{array}{c} p \\ q \end{array} \right) \Leftrightarrow
    p \leq max\big\{ \left\lfloor \frac{q}{b} \right\rfloor a, q + \left\lceil \frac{q}{b} \right\rceil (a - b) \big\}
\end{equation}

According to Theorem~3 in~\cite{bernat2001weakly} we have $\hw \equiv \wh$. Hence, replacing the variables for our case, we get:
\begin{multline}
    \left( \begin{array}{c} h_i \\ w_i + h_i \end{array}\right) \preccurlyeq \left( \begin{array}{c} K_i - m_i \\ K_i \end{array} \right) \Leftrightarrow \\
    K_i - m_i \leq max\left\{ \left\lfloor \frac{K_i}{ w_i + h_i} \right\rfloor h_i, K_i - \left\lceil \frac{K_i}{ w_i + h_i} \right\rceil  w_i \right\}
    \label{equ:theorem-harder-constraint}
\end{multline}

For proving the theorem, we need to show that $K_i - m_i$ is always less or equal than the following terms:
\begin{equation}
    \left\lfloor \frac{K_i}{ w_i + h_i} \right\rfloor h_i
    \label{equ:theorem-proof-term-1}
\end{equation}
\begin{equation}
    K_i - \left\lceil \frac{K_i}{ w_i + h_i} \right\rceil  w_i
    \label{equ:theorem-proof-term-2}
\end{equation}

We prove this theorem separately for low-tolerance and high-tolerance tasks.
Summarizing, the steps are the following:
\begin{enumerate}
    \item We show Equation~\eqref{equ:theorem-proof-term-1} is greater or equal than Equation~\eqref{equ:theorem-proof-term-2} for high-tolerance tasks.
    \item We verify that Equation~\eqref{equ:theorem-proof-term-1} is greater or equal than $K_i - m_i$ for high-tolerance tasks.
    \item We show Equation~\eqref{equ:theorem-proof-term-1} is also greater or equal than Equation~\eqref{equ:theorem-proof-term-2} for low-tolerance tasks.
    \item Then, for simplicity, we start by verifying that Equation~\eqref{equ:theorem-proof-term-2} is greater or equal than $K_i - m_i$ for low-tolerance tasks.
    Showing this last is valid, we conclude that for low-tolerance tasks Equation~\eqref{equ:theorem-proof-term-1} is also greater or equal than $K_i - m_i$.
\end{enumerate}

{\bf Step 1.}
From Lemma~\ref{lem:w-h-particular-values}, we know that $h_i = 1$ for high-tolerance tasks.
Replacing $h_i = 1$ and assuming that Equation~\eqref{equ:theorem-proof-term-1} $\geq$ Equation~\eqref{equ:theorem-proof-term-2}, we get:
\begin{equation}
    \left\lfloor \frac{K_i}{ w_i + 1} \right\rfloor \geq K_i - \left\lceil \frac{K_i}{ w_i + 1} \right\rceil  w_i
    \label{equ:step-1}
\end{equation}

We consider the following relation between a real number $x$ and a integer number $n$ for removing the floor.
\begin{equation}
    \lfloor x \rfloor \geq n \Leftrightarrow x \geq n
    \label{equ:real-integer-floor-relation}
\end{equation}

By removing the floor and then clearing $K_i$ in Inequation~\eqref{equ:step-1}, we get the assumption Equation~\eqref{equ:theorem-proof-term-1} $\geq$ Equation~\eqref{equ:theorem-proof-term-2} holds.
\begin{multline*}
    \frac{K_i}{ w_i + 1} \geq K_i - \left\lceil \frac{K_i}{ w_i + 1} \right\rceil  w_i \\
    K_i \geq \left( K_i - \left\lceil \frac{K_i}{ w_i + 1} \right\rceil  w_i \right) ( w_i + 1) \\
    \left\lceil \frac{K_i}{ w_i + 1} \right\rceil ( w_i + 1) \geq K_i
\end{multline*}

{\bf Step 2.}
We need to prove:
\begin{equation}
    \left\lfloor \frac{K_i}{ w_i + 1} \right\rfloor \geq K_i - m_i
    \label{equ:step-2}
\end{equation}

Using Relation~\ref{equ:real-integer-floor-relation} for removing the floor:
\begin{equation*}
    \frac{K_i}{ w_i + 1} \geq K_i - m_i
\end{equation*}
Next, clearing $ w_i$ and replacing it for its value for high-tolerance tasks, i.e. $ w_i = \left\lfloor \frac{m_i}{K_i - m_i} \right\rfloor$:
\begin{equation*}
     w_i \leq \frac{m_i}{K_i - m_i} \Longrightarrow
    \left\lfloor \frac{m_i}{K_i - m_i} \right\rfloor \leq \frac{m_i}{K_i - m_i}
\end{equation*}
Inequation~\eqref{equ:step-2} holds and we prove Equation~\eqref{equ:theorem-harder-constraint} for high-tolerance tasks.

{\bf Step 3.}
From Lemma~\ref{lem:w-h-particular-values}, we know that $ w_i = 1$ for low-tolerance tasks.
Replacing $ w_i = 1$ and assuming that Equation\eqref{equ:theorem-proof-term-1} is greater or equal than Equation\eqref{equ:theorem-proof-term-2}, we get:
\begin{equation}
    \left\lfloor \frac{K_i}{1 + h_i} \right\rfloor h_i \geq K_i - \left\lceil \frac{K_i}{1 + h_i} \right\rceil
    \label{equ:step-3}
\end{equation}

For converting from ceiling to floor, we use the following identity:
\begin{equation}
    \left\lceil \frac{a}{b} \right\rceil = \left\lfloor \frac{a + b - 1}{b} \right\rfloor
    \label{equ:ceiling-floor-conversion}
\end{equation}

Hence:
\begin{equation*}
    \left\lfloor \frac{K_i}{1 + h_i} \right\rfloor h_i \geq K_i - \left\lfloor \frac{K_i + h_i}{1 + h_i} \right\rfloor \Longrightarrow
    \left\lfloor \frac{K_i + h_i}{1 + h_i} \right\rfloor \geq K_i - \left\lfloor \frac{K_i}{1 + h_i} \right\rfloor h_i
\end{equation*}

Using Relation~\ref{equ:real-integer-floor-relation} for removing the floor in $\left\lfloor \frac{K_i + h_i}{1 + h_i} \right\rfloor$ and then clearing $K_i$:
\begin{multline*}
    \frac{K_i + h_i}{1 + h_i} \geq K_i - \left\lfloor \frac{K_i}{1 + h_i} \right\rfloor h_i \\
    K_i + h_i \geq \left( K_i - \left\lfloor \frac{K_i}{1 + h_i} \right\rfloor h_i \right) (1 + h_i) \\
    (1 + h_i) \left\lfloor \frac{K_i}{1 + h_i} \right\rfloor + 1 \geq K_i
\end{multline*}
We see that assumption Equation~\eqref{equ:theorem-proof-term-1} $\geq$ Equation~\eqref{equ:theorem-proof-term-2} holds.

{\bf Step 4.}
As we mentioned before, here we start by showing that Equation~\eqref{equ:theorem-proof-term-2} is greater or equal than $K_i - m_i$:
\begin{equation}
    K_i - \left\lceil \frac{K_i}{1 + h_i} \right\rceil \geq K_i - m_i \Longrightarrow
    m_i \geq \left\lceil \frac{K_i}{1 + h_i} \right\rceil
    \label{equ:step-4}
\end{equation}

We now use the following relation between a real number $x$ and a integer number $n$ to remove the ceiling from Inequation~\eqref{equ:step-4}.
\begin{gather*}
    \lceil x \rceil \leq n \Leftrightarrow x \leq n \\
    m_i \geq \left\lceil \frac{K_i}{1 + h_i} \right\rceil \Longrightarrow
    m_i \geq \frac{K_i}{1 + h_i}
\end{gather*}

Next step is clearing $h_i$ and replacing it for its value, i.e.\ $h_i = \left\lceil \frac{K_i - m_i}{m_i} \right\rceil$:
\begin{equation*}
    h_i \geq \frac{K_i - m_i}{m_i} \Longrightarrow
    \left\lceil \frac{K_i - m_i}{m_i} \right\rceil \geq \frac{K_i - m_i}{m_i}
\end{equation*}

We see the assumption Equation~\eqref{equ:theorem-proof-term-2} $\geq K_i - m_i$ holds and since also Equation~\eqref{equ:theorem-proof-term-1} $\geq$ Equation~\eqref{equ:theorem-proof-term-2}, we get also Equation~\eqref{equ:theorem-proof-term-1} $\geq K_i - m_i$.
This proves Equation~\eqref{equ:theorem-harder-constraint} for low-tolerance tasks which finalizes the proof of Theorem~5 of \cite{bernat2001weakly} for $\hw \preccurlyeq   \big( \begin{smallmatrix} K_i-m_i \\ K_i \end{smallmatrix}\big)$.

\end{document}